\documentclass[a4paper]{article}
\usepackage{latexsym,amsmath,amssymb,enumerate,mathrsfs,amsthm}
\usepackage{graphicx}
\usepackage{authblk}
\usepackage{subcaption}

\numberwithin{equation}{section}

\newcommand{\alg}[1]{\mathcal{#1}}

\newcommand{\mc}[1]{\mathcal{#1}}
\newcommand{\mb}[1]{\mathbb{#1}}
\newcommand{\mr}[1]{\mathrm{#1}}

\newcommand{\qd}[1]{\mathcal{D}(#1)}
\newcommand{\qdg}{\qd{G}}

\newcommand{\str}{\operatorname{star}}
\newcommand{\plaq}{\operatorname{plaq}}

\newcommand{\supp}{\operatorname{supp}}

\newcommand{\bonds}{\Gamma}

\newcommand{\wlim}{\operatorname{w-lim}}
\newcommand{\ket}[1]{\ensuremath{\left|#1\right\rangle}}
\newcommand{\eps}{\epsilon}

\newcommand{\ovl}[1]{\ensuremath{\mkern 1.0mu\overline{\mkern-1.0mu #1\mkern-1.0mu }\mkern 1.0mu }}
\newcommand{\inv}[1]{\ovl{#1}}
\theoremstyle{plain}
\newtheorem{theorem}{Theorem}
\numberwithin{theorem}{section}
\newtheorem{lemma}[theorem]{Lemma}
\newtheorem{proposition}[theorem]{Proposition}

\theoremstyle{definition}
\newtheorem{definition}[theorem]{Definition}

\theoremstyle{remark}
\newtheorem{remark}[theorem]{Remark}

\title{Haag duality for Kitaev's quantum double model for abelian groups}

\author[1]{Leander Fiedler\thanks{\texttt{leander.fiedler@itp.uni-hannover.de}}}
\author[1]{Pieter Naaijkens\thanks{\texttt{pieter.naaijkens@itp.uni-hannover.de}}}
\affil[1]{Institut f\"ur Theoretische Physik\\Leibniz Universit\"at Hannover, Germany}

\begin{document}
\maketitle
\begin{abstract}
	We prove Haag duality for conelike regions in the ground state representation corresponding to the translational invariant ground state of Kitaev's quantum double model for finite abelian groups. This property says that if an observable commutes with all observables localised outside the cone region, it actually is an element of the von Neumann algebra generated by the local observables inside the cone. This strengthens locality, which says that observables localised in disjoint regions commute.

As an application we consider the superselection structure of the quantum double model for abelian groups on an infinite lattice in the spirit of the Doplicher-Haag-Roberts program in algebraic quantum field theory. We find that, as is the case for the toric code model on an infinite lattice, the superselection structure is given by the category of irreducible representations of the quantum double.
\end{abstract}

\section{Introduction}
Kitaev's quantum double model for finite groups is a spin model on a 2D lattice that exhibits anyonic excitations~\cite{MR1951039}. One of its main features is that it has certain topological properties: the ground space degeneracy depends on the topology of the underlying lattice. In addition, the model has (quasi-)particle excitations with braid (anyonic) statistics. This can be exploited to perform quantum computations. In fact, for certain groups it allows even for \emph{universal} quantum computation~\cite{PhysRevA.67.022315,PhysRevA.69.032306}. The computations are made possible by the braid statistics of the anyons which are encoded in the superselection structure of the model.

It turns out that even on a topologically trivial lattice, such as a square lattice on the plane, the excitations with anyonic statistics exist. One way to recover the properties of these excitations is by doing a Doplicher-Haag-Roberts (DHR) type analysis of the superselection sectors~\cite{MR0297259,MR0334742}. This has been carried out for the toric code in ~\cite{toricendo,haagdtoric}. The toric code corresponds to the choice of $G=\mathbb{Z}_2$ in Kitaev's quantum double model. Here we extend these results to general finite abelian groups $G$. In particular, we show that algebras of observables localised in cone-like regions fulfil Haag duality in the vacuum representation. This means that in the vacuum representation observables which commute with all observables outside the cone are exactly those which can be approximated (in the weak operator topology) by operators localised inside the cone. More precisely, suppose $\Lambda$ is a cone and $\alg{A}(\Lambda)$ is the algebra of quasi-local observables localised in $\Lambda$. Then we have in the ground state representation $\pi_0$ of the quantum double model, that $\pi_0(\alg{A}(\Lambda))'' = \pi_0(\alg{A}(\Lambda^c))'$, where the prime denotes taking the commutant, and $\Lambda^c$ is the set of all sites in the complement of the cone $\Lambda$. Note that one of the inclusions readily follows from locality, the other one is non-trivial and is what the larger part of this paper is devoted to.

For the proof of Haag duality we follow the ideas introduced in~\cite{haagdtoric}. In particular, we first show that cone algebras leave certain subspaces of the vacuum Hilbert space invariant. These subspaces can be shown to be generated by the self-adjoint parts of the cone algebra and the algebra associated to the complement of the cone. Using a result by Rieffel and van Daele \cite{RiDa:1975} we can then conclude Haag duality for cone algebras. In essence the proof relies on a thorough understanding of the ground state, or rather, the full excitation spectrum of the model. This allows us to get a precise understanding of the Hilbert spaces describing all pairs of excitations in a certain region of the system. A precise understanding of how these states can be obtained by acting with local operators on the ground state vector allows us to use the aforementioned theorem by Rieffel and van Daele.

As a consequence of Haag duality and a property of the ground state we obtain the approximate\footnote{In earlier work~\cite{toricendo,haagdtoric} we called this the \emph{distal} split property. However, we feel \emph{approximate} is more appropriate, since we will only need to assume a small separation of two regions.} split property, which implies that if we have two cones that are removed from each other sufficiently far, then we can prepare normal states on these two cones independently~\cite{MR735338}. In this sense it is a form of statistical independence of the two cone regions. This is no longer true if we take a cone $\Lambda$ and its complement. In that case, the split property does not hold any more and one cannot find \emph{normal} product states on the two regions, and they are not independent in the strong sense~(c.f.~\cite{MR984150}).

Another application of Haag duality that we consider is the analysis of the superselection structure of the model for finite abelian groups. We show that in this case the superselection structure is described by conelike localised endomorphisms by explicitly constructing such endomorphisms that describe a single excitation. In this way we can show that the superselection structure (including the braiding and fusion roles) is described by finite dimensional representations of Drinfeld's quantum double $\qdg$ of the underlying group. This resembles analogue results for Kitaev's toric code model~\cite{toricendo}. We do this by constructing states describing a single charge. These are obtained by creating a pair of excitations from the ground state, and move one of these to infinity. Because of the topological properties of the quantum double model, the direction in which we do this cannot be observed. It follows that the corresponding representations satisfy a superselection criterion: they are irreducible representations that are unitarily equivalent to the ground state representation, but \emph{only} when one restricts to observables localised outside a cone. This resembles the Buchholz-Fredenhagen criterion in algebraic quantum field theory~\cite{MR660538}. Using Haag duality we can then restrict to endomorphisms of the observables, and do the DHR superselection theory~\cite{MR0297259,MR0334742}.

The paper is organised as follows. In Section~\ref{sec:qdouble} we review the geometric and algebraic setting of the quantum double model for finite groups and introduce necessary notation. We also recall the main properties of the excitation spectrum of the model. This is then used in Section~\ref{sec:ground} to show that there is a unique translational invariant ground state in our setting. Section~\ref{sec:haagd} contains the main result of this paper: the proof of Haag duality. The next two sections concern the approximate split property and an analysis of the superselection structure of abelian quantum double models. We end with an outlook on the extension to non-abelian groups of our results.

\section{The quantum double model}\label{sec:qdouble}
We start with recalling the setting of the quantum double model for finite groups. Most of the result in this section are not new, but since the notation and properties we introduce here play an important role in the main part of this paper, we recall the essentials to make the paper more self-contained. For a more detailed introduction we refer to \cite{MR1951039} and \cite{PhysRevB.78.115421}. We will mainly follow the notation of~\cite{PhysRevB.78.115421}.

\begin{figure}
    \centering
    \includegraphics[width=0.25\linewidth]{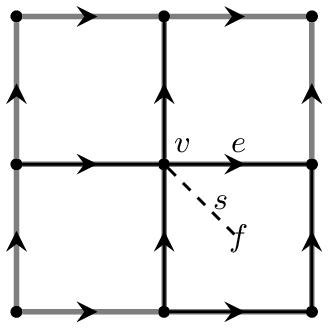}
    \caption{Geometric setting of the quantum double model. The black arrows on the edges indicate their orientation. $v$, $e$ and $f$ are a vertex, edge and face, respectively, and $s=(v,f)$ is a site.
    The star at $v$ is given by the four black edges connecting to $v$. The plaquette at $f$ is defined by the four edges surrounding $f$.}
    \label{fig:site}
\end{figure}

Consider a square lattice $\mb Z^2$ and let $G$ be any finite group.\footnote{We later specialise to \emph{abelian} groups, but for the definition of the model this is immaterial.} Vertices of the lattice are denoted by $v$. Between nearest-neighbour vertices there are oriented edges $e$. The set of all these edges (or bonds) is called $\Gamma$. For simplicity we fix the orientation of the edges as in Figure~\ref{fig:site}: the edges point either right or up.
If $e$ is an edge,we write $\ovl e$ for the edge with the opposite orientation. Faces of the lattice are denoted by $f$. Note that faces can be identified with vertices in the dual lattice, and similarly vertices in the lattice correspond to faces in the dual lattice. The edges in the dual lattice connect two faces of the lattice.
They inherit an orientation from the edges in the lattice in the following way. Given an edge $e$ in $\Gamma$ then its orientation fixes a notion for the neighbouring faces to lie to the ``left'' or to the ``right'' of $e$ in direction of the orientation.
A dual edge $\ovl e$ is then oriented in such a way that it point from the face to the right to the face to the left of the associated oriented edge in $\Gamma$.
A \emph{site} $s$ will mean a tuple $s=(v,f)$ of a vertex and a neighbouring face $f$. Finally, we refer to the four edges enclosing a face as a plaquette (notation: $\plaq(f)$), and the four edges ending or starting at a vertex $v$ as a star, notation $\str(v)$.

To each edge $e$ of the lattice we associate a Hilbert space $\mc H_e$ with a basis labelled by the elements of $G$.\footnote{One can in fact regard it as the group $C^*$-algebra $C[G]$. More generally, this can be done for any finite dimensional Hopf-$*$ algebra (c.f. the remark on p.13 of~\cite{MR1951039}).} The orthonormal basis in $\mc H_e$ is denoted by $\ket g$ with $g\in G$.
For any edge $e\in\Gamma$ denote $\mc A(\{e\}):=\mc B(\mc H_e)$ for the algebra of observables acting on this edge. Similarly for any set $\mc O\in\mc P_f(\Gamma)$, where $\mc{P}_f(\Gamma)$ is the set of all \emph{finite} subsets of $\Gamma$, the local algebras are given by $\mc A(\mc O):=\bigotimes_{e\in\mc O}\mc A(\{e\})$. We will also write $\mc A_e$ for $\mc A(\{e\})$. If $A \in \alg{A}(\mc{O})$ we say that $A$ is \emph{localised} in $\mc{O}$.

This construction gives rise to an isotonous net $\mc O\mapsto\mc A(\mc O)$ of $C^*$-algebras. The corresponding embedding $*$-isomorphisms are given by the natural embedding provided by the tensor product structure. That is, if $\mc{O}_1 \subset \mc{O}_2$ it is given by the map $\iota_{\mc{O}_1 \mc{O}_2}$ defined by $A \mapsto A \otimes I_{\mc{O}_2 \setminus \mc{O}_1}$. The inductive limit of this net is called the \emph{quasilocal algebra} and denoted by $\alg{A}$. It is the norm closure of the $*$-algebra all observables localised in \emph{finite} regions. Similarly one can define for \emph{infinite} sets $\Lambda \subset \Gamma$ the algebra $\alg{A}(\Lambda) \subset \Lambda$ as the norm closure of $\bigcup_{\mc{O} \subset \Lambda} \alg{A}(\mc{O})$, where the union is over \emph{finite} subsets.

An important part of the model's structure is most easily explained in terms of certain operators associated to triangles on the lattice. We recall the main definitions here. A \emph{direct triangle} $\tau$ can be thought of connecting a face with two neighbouring vertices which are connected by an edge. More specific, consider a site $s=(v,f)$ can be thought of as a line connecting a vertex $v$ with a face $f$.
A direct triangle $\tau$ is then a tuple $(s_1,s_2,e)$ of the sites $s_1,s_2$ and an edge $e$ or its inverse such that the tuple lists the sides of the triangle $\tau$ in clockwise order.Similarly, a \emph{dual triangle} $\tau'$ connects two neighbouring faces with a vertex over some dual edge. Again $\tau'$ is given by a tuple $(s_1,s_2,e)$ where $s_1,s_2$ are sites and $e$ is an edge or its dual edge such that the tuple lists the sides of $\tau'$ in \emph{counterclockwise} order.
Given a triangle $\tau=(s_1,s_2,e)$ denote $\partial_0\tau:=s_1$ and $\partial_1\tau:=s_2$. Two triangles are said to \emph{overlap} if and only if the corresponding edges intersect.
Any triangle inherits an orientation by the edge in the tuple. Note that it can either coincide with the orientation given by the lattice (or dual lattice) $\Gamma$, or be anti-parallel.

Given a direct triangle $\tau$ and an element $h\in G$ we can now associate an operator $T_\tau^h\in\mc A_e$ by
\begin{equation*}
    T_\tau^h\ket k =\left\{
        \begin{array}{ll}
            \delta_{h,k}\ket k, &   \textnormal{ if }\tau\textnormal{ is oriented parallel to }\Gamma\\
        \delta_{\ovl h,k}\ket k,    &   \textnormal{ else}
    \end{array}
    \right.
    ,\ k\in G.
\end{equation*}
For a dual triangle $\tau'$ we set for any element $g\in G$
\begin{equation*}
    L_{\tau'}^g\ket k =\left\{
        \begin{array}{ll}
            \ket{gk}, &   \textnormal{ if }\tau\textnormal{ is oriented parallel to }\Gamma\\
            \ket{k\ovl g},    &   \textnormal{ else}
    \end{array}
    \right.
    ,\ k\in G.
\end{equation*}
Here and in the remainder of the paper we will use the notation $\ovl g$ for the inverse group element of $g\in G$, to keep the sub and superscripts in the formula more readable.
If $\tau$ and $\tau'$ overlap the corresponding triangle operators act on the same edge $e$ and one can verify that the operators $L_{\tau'}^gT_\tau^h,\ h,g\in G$ are matrix units spanning $\mc A_e$.

A crucial role is played by operators that act along a \emph{ribbon}. A ribbon $\rho$ is given by a tuple $(\tau_1,\tau_2,\dots,\tau_n)$ of pairwise non-overlapping triangles such that $\partial_1\tau_i=\partial_0\tau_{i+1}, i=1,\dots,n-1$. We set $\partial_0\rho:=\partial_0\tau_1$ and $\partial_1\rho:=\partial_1\tau_n$.
A ribbon $\rho$ is called is called \emph{closed} if $\partial_1\rho=\partial_0\rho$. Given a ribbon $\rho$ and group elements $g,h \in G$ an associated ribbon operator $F_\rho^{g,h}$ is defined recursively as follows: let $\tau$ be a direct, $\tau'$ a dual triangle and $\epsilon$ the trivial ribbon. In this case the ribbon operators are defined as
\begin{align*}
    F_\eps^{g,h}    & := I,  &   F_\tau^{g,h}    & := T_\tau^g,   &   F_{\tau'}^{g,h} &:= \delta_{e,h}L_{\tau'}^g,
\end{align*}
where $e\in G$ is the unit element.
If $\rho$ is any ribbon, we can decompose it into two possibly smaller ribbons $\rho_1$ and $\rho_2$, and write $\rho=\rho_1\rho_2$. The ribbon operator on $\rho$ is then defined recursively in terms of the ribbon operators on the smaller ribbons by
\begin{align}
    F_\rho^{g,h}:=\sum_{k\in G}F_{\rho_1}^{g,k}F_{\rho_2}^{\ovl hgh,\ovl kh}.
	\label{eq:ribdecompose}
\end{align}
It can be checked that this is consistent and independent of the partition~\cite{PhysRevB.78.115421}. We will sometimes refer to equation~\eqref{eq:ribdecompose} as the \emph{ribbon decomposition rule}.
The algebra generated by the ribbon operators acting along a ribbon $\rho$ will be denoted by $\mc F_\rho$.

The commutation relations for ribbon operators associated to some ribbon $\rho$ are given by $F_\rho^{g,h}F_\rho^{k,l} = \delta_{h,l}F_\rho^{gk,l}$ and $[F_\rho^{g,h},F_{\rho'}^{k,l}]=0$ if $\rho\cap\sigma=\emptyset$.
The case where $\rho$ and $\sigma$ overlap at some site will be discussed later in the context of braiding (see Section~\ref{sec:rib_comm}). Finally, the adjoint is given by $(F_\rho^{h,g})^* = F_\rho^{\inv{h},g}$.

Given a site $s$ there are two distinct closed ribbons that start and end at $s$, namely the smallest closed ribbon $\beta_s$ that consists just of direct triangles and the smallest closed ribbon $\alpha_s$ consisting only of dual triangles.
For $g,h\in G$ we set $A_s^g:=F_{\alpha_s}^{g,e}$ and $B_s^h:=F_{\beta_s}^{e,\ovl h}$ and define the \emph{star} and \emph{plaquette} operators by
\begin{align}
    A_s &:= \frac1{|G|}\sum_{g\in G}A_s^g    &   B_s    &:= B_s^e.
    \label{eq:star_plaq}
\end{align}
The definition definition of $A_s$ depends only on the vertex the site is located at and $B_s$ depends only on the face at $s$. The name \emph{star operator} can be explained by noting that it acts on the edges of $\str(s)$. Similarly, the plaquette operator acts on the corresponding plaquette.

There is another convenient description of the operators $A_s^g$ and $B^h_s$. It can be obtained by choosing a basis vector in the tensor product of the Hilbert spaces on the edges, and describing the action on this basis vector. As an example the action of $A_s^g$ is visualized the following diagram~\eqref{eq:starops}:
\begin{equation}
\label{eq:starops}
\raisebox{-0.5\height}{\includegraphics{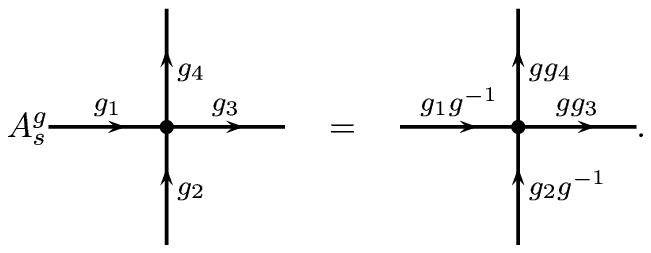}}
\end{equation}
In a similar way one can see that $B_s^h$ is a projection. In particular, choose a basis element corresponding to a choice of group elements $g_1, \dots g_4$, labelled by following the edges around the plaquette in an anticlockwise direction, starting at the vertex $v$ of the site $s$. The action of $B^h_s$ on this vector is the identity if $\sigma(g_1) \sigma(g_2) \sigma(g_3) \sigma(g_4) = h$, and zero otherwise. Here $\sigma(g_i) = g_i$ if the direction of the edge matches the anticlockwise path, and $\inv{g}_i$ otherwise. The product of the group elements is also called the \emph{flux} through the plaquette.

Let $s$ be a site. Using the definition above, it is not so difficult to work out the commutation relations for operators $A_s^g, B_s^h$ acting on the site $s$. One finds
\begin{equation}
	A^g_s A^{g'}_s = A_s^{g g'}, \quad B^h_s B^{h'}_s = \delta_{h,h'} B^h_s,\quad A^g_s B^h_s = B^{g h g^{-1}}_s A^g_s.
\label{eq:agbh}
\end{equation}
In particular this shows that for any pair of sites $s,s'$ the plaquette and star operators commute, i.e. $[A_s,B_{s'}]=0=[A_s,A_{s'}]=0=[B_s,B_{s'}]$.

\begin{remark}
	The operators $A_s^g$ and $B_s^h$ generate a finite dimenional algebra, that is in fact isomorphic to Drinfeld's \emph{quantum double}~\cite{MR934283} of the group algebra $\mathbf{C}[G]$ regarded as a Hopf algebra. We write $\qdg$ for this algebra. This explains the name ``quantum double model''. The quantum double has been very well studied, and many of the properties that we will need in this paper can be traced back to the representation theory of $\qdg$. Good introductions can be found in, for example, Ref.~\cite{MR1234107} for the quantum double and its representations in the context of $C^*$-algebras, or the textbook~\cite{MR1321145} for a more algebraic approach in the language of category theory.
\end{remark}

With this notation we can introduce the dynamics of the quantum double model. Recall that dynamics can be specified by local Hamiltonians, satisfying certain conditions that ensure that they lead to a time evolution on the entire quasi-local algebra of observables $\alg{A}$~\cite{MR1441540}. These local Hamiltonians can be defined in terms of the operators $A_s^g$ and $B_s^h$ introduced above, or rather the sum $A_s$ and the projection $B_s$. Note that these operators mutually commute, even if they both act on the same site. We sometimes write $B_f$ or $A_v$, where $f$ is a face and $v$ a vertex, instead of $s = (v,f)$. Note that this does not lead to ambiguities.

Concretely, let $\Lambda \in \mc{P}_f(\bonds)$. Then the corresponding local Hamiltonian is defined by
\[
H_{\Lambda} = - \sum_{\str(s) \subset \Lambda} A_s - \sum_{\plaq(s) \subset \Lambda} B_s.
\]
The summation is over all stars and all plaquettes (faces) whose bonds are completely contained in $\Lambda$. We will later see that the ground state is a stabilizer state, that is stabilized by each $A_s$ and $B_p$, and we can see the Hamiltonian as implementing an energy penalty for violation of the ``constraints'' that $A_s \Omega = B_s \Omega = \Omega$ for a ground state (as we will see later).

Ribbon operators are interpreted as creating excitations at the ending sites of the ribbon. This interpretation is strengthened by the commutation relations with star and plaquette operators (see \cite{PhysRevB.78.115421})
\begin{equation}\label{eqn:ABF}
    \begin{array}{r@{\ =\ }lr@{\ =\ }l}
        A_{s_0}^kF_\rho^{h,g}   &  F_\rho^{kh\ovl k,kg}A_{s_0}^k   &   A_{s_1}^kF_\rho^{h,g}   &  F_\rho^{h,g\ovl k}A_{s_1}^k\\
        B_{s_0}^kF_\rho^{h,g}   &  F_\rho^{h,g}B_{s_0}^{kh}    &   B_{s_1}^kF_\rho^{h,g}   &  F_\rho^{h,g}B_{s_1}^{\ovl g\ovl hgk}
    \end{array}
\end{equation}
for a ribbon $\rho$ with $s_i=\partial_i\rho,\ i=0,1$ and $g,h,k\in G$. On the other hand the stars and plaquettes at sites different from $s_1$ and $s_2$ commute with $\rho$. Hence if we act with a ribbon operator on the ground state, some of the constraints in the Hamiltonian will be violated.

For our purposes it will be convenient to consider a different basis for the space of ribbon operators acting on a ribbon $\rho$. Let $C=\{c_1,\dots,c_n\}$ be a conjugacy class of $G$, $r\in C$ some representative and $\pi$ an irreducible unitary representation of $Z_G(r)$, the centraliser of $r$ in $G$. Choose elements $q_1,\dots,q_n$ such that $c_i=q_ir\ovl q_i$ for $i=1,\dots,n$ and set
\begin{align}
    F_\rho^{C,\pi,i,i',j,j'}:=\sum_{z\in Z_G(r)}\ovl{\pi_{j,j'}(z)}F_\rho^{\ovl c_i,q_ir\ovl q_{i'}}
    \label{eq:charge_ribbons}
\end{align}
where $j,j'\in\{1,\dots,|\pi|\}$ label the matrix elements of $\pi$ and $i,i'\in\{1,\dots,n\}$. This relates the ribbon operators to irreducible representations of the quantum double $\qdg$ of the group~\cite{MR1128130}. It can be shown (see \cite{PhysRevB.78.115421}), that in case $\rho$ consists of both (direct and dual) types of triangles then these operators form a basis of $\mc F_\rho$, the algebra generated by the ribbon operators at $\rho$.

Note that if the group $G$ is abelian, these definitions somewhat simplify, essentially because we only have to deal with one dimensional representations. In that case, the centraliser is simply $G$, the conjugacy classes are single elements, and the irreducible representations are characters $\chi$ of $G$. We simply write $F^{\chi,c}_\rho$ in that case, that is,
\begin{equation}
	\label{eq:abribbon}
	F^{\chi,c}_\rho = \sum_{g \in G} \overline{\chi}(g) F^{\inv{c}, g}_\rho.
\end{equation}
It is not difficult to check that $F_\rho^{\chi_1,c} F_\rho^{\chi_2,d} = F_{\rho}^{\chi_1 \chi_2, cd}$, where $\chi_1 \chi_2$ is the pointwise product of $\chi_1$ and $\chi_2$.

There is another useful property that is valid for these operators for abelian $G$, but not in general: if we decompose a ribbon $\rho$ into two parts $\rho_1$ and $\rho_2$, the corresponding ribbon operator is just the product of the ribbon operators acting along the smaller ribbons.
\begin{lemma}
	\label{lem:abdecompose}
Let $\chi$ be a character and $c \in G$ for some finite abelian group $G$. Suppose that $\rho = \rho_1 \rho_2$ is a ribbon. Then $F^{\chi,c}_\rho = F^{\chi,c}_{\rho_1} F^{\chi,c}_{\rho_2}$.
\end{lemma}
\begin{proof}
	With the help of equation~\eqref{eq:ribdecompose} we find that
\[
	F_\rho^{\chi,c} = \sum_{g,k \in G} \overline{\chi}(g) F^{c,k}_{\rho_1} F^{c,\inv{k} g}_{\rho_2} = \sum_{g,k \in G} \overline{\chi}(k g) F^{c,k}_{\rho_1} F^{c,g}_{\rho_2}
					  = F_{\rho_1}^{\chi,c} F_{\rho_2}^{\chi,c},
\]
where we made the substitution $g \mapsto kg$ in the second equality, and used that $\overline{\chi}(kg) = \overline{\chi}(k) \overline{\chi}(g)$.
\end{proof}

\subsection{Properties of ribbon operators}\label{sec:proprib}
For later use we list some properties of ribbon operators that we will need later. In particular, we are interested in the question how the action of these operators on the ground state depend on the ribbon itself. As will be outlined in Section~\ref{sec:ground}, in the present situation there is a unique translational invariant ground state $\omega_0$. The corresponding GNS representation will be denoted by $(\pi_0, \Omega, \mc{H}_0)$. If we talk about ``the ground state'' or ``ground state vector'', we will always mean the translational invariant ground state $\omega_0$ (resp. the GNS vector $\Omega$). Since $\alg{A}$ is an inductive limit of simple algebras, it is simple, hence $\pi_0$ is a faithful representation. To simplify notation we therefore often write simply $A$ for $\pi_0(A)$. An essential fact in proving the properties below is that the ground state vector $\Omega$ has the property that $A^g_s \Omega = \Omega$ and $B^h_s \Omega = \delta_{h,e} \Omega$ for any site $s$ and group elements $g,h \in G$.

    As it turns out the action of a ribbon operator on the ground state only depends on the sites connected by the ribbon and not on the connecting ribbon itself. This allows to deform ribbons by fixing its endpoints and changing the shape in between. In the following $G$ denotes any finite group.
	\begin{lemma}\label{lem:indeprib}
        Let $\rho,\rho',\sigma,\sigma'$ be ribbons with $\partial_i\rho=\partial_i\rho'$ and $\partial_i\sigma=\partial_i\sigma',\ i=0,1$. Then for all $A\in\mc A$ and all $g,h,k,l\in G$ it holds
        \begin{align*}
            \omega_0(F_\rho^{h,g}AF_\sigma^{l,k})=\omega_0(F_{\rho'}^{h,g}AF_{\sigma'}^{l,k}).
        \end{align*}
    \end{lemma}
    \begin{proof}
        In \cite{PhysRevB.78.115421} it is shown that for ribbons $\rho,\rho'$ as above the ribbon operators $F_\rho^{h,g},F_{\rho'}^{h,g}$ map the ground state vector $\Omega$ to the same image, i.e. $F_\rho^{h,g}\Omega=F_{\rho'}^{h,g}\Omega$.
        Hence by noting that $\omega_0(A)=\langle\Omega ,A\Omega\rangle,\,A\in\mc A$ and $(F_\rho^{h,g})^*=F_\rho^{\ovl h,g}$ the claim follows.
    \end{proof}
    We refer to $\rho'$ and $\sigma'$ as \emph{deformations} of $\rho$ and $\sigma$. A more detailed definition and description can be found in \cite{PhysRevB.78.115421}.
    An example of a deformation of a ribbon is given in Figure~\ref{fig:rib_def}.
    \begin{figure}
        \centering
        \includegraphics[width=0.5\linewidth]{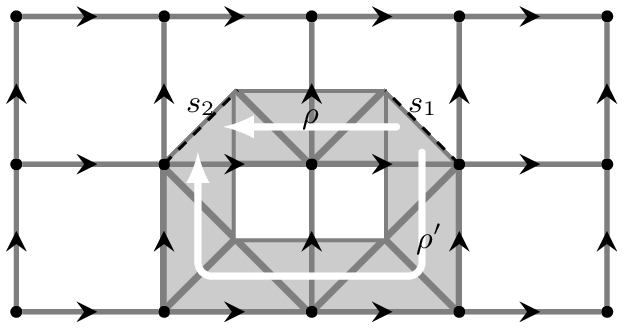}
        \caption{A deformation $\rho'$ of a ribbon $\rho$ connection sites $s_1$ and $s_2$. The white arrows indicates the orientation of the ribbons.}
        \label{fig:rib_def}
    \end{figure}
    Later on we need to connect to ending sites of a ribbon $\rho$ with another ribbon $\ovl\rho$ such that $\rho\ovl\rho$ is a closed ribbon. The way ribbon operators of $\rho$ and $\ovl\rho$ relate to each other is given by the following lemma:
    \begin{lemma}\label{lem:ribinv}
        Let $\rho,\ovl\rho,\sigma,\ovl\sigma$ be ribbons with $\partial_i\rho=\partial_{1-i}\ovl\rho$ and $\partial_i\sigma=\partial_{1-i}\ovl\sigma,\ i=0,1$. Then for all $A\in\mc A$ and all $g,h,k,l\in G$ it holds
        \begin{align*}
            \omega_0(F_\rho^{h,g}AF_\sigma^{l,k})=\omega_0(F_{\ovl\rho}^{\ovl h,\ovl g}AF_{\ovl\sigma}^{\ovl l,\ovl k}).
        \end{align*}
    \end{lemma}
	For the proof we refer to \cite{PhysRevB.78.115421}.
   The ribbons $\ovl\rho$ and $\ovl\sigma$ are referred to as \emph{inversions} of $\rho$ and $\sigma$. We can always choose such an inversion of a ribbon.

From now on assume that $G$ is abelian. Then the irreducible representations of $\mc D(G)$ are given by elements $c\in G$ and characters $\chi:G\to\mb C$, as mentioned above.
\begin{lemma}\label{lem:ab_tool_3}
    Let $\rho\subset\Lambda$ be a ribbon, $c,k\in G$, $\chi$ an irreducible character and $s$ the starting site of $\rho$. Then
    \begin{align*}
        (F_\rho^{\chi,c})^*   &=  F_\rho^{\ovl\chi,\ovl c}\\
        A_s^kF_\rho^{\chi,c}  &=  \sum_{g\in G}\ovl\chi(g)A^kF^{\ovl c,g} = \chi(k)F_\rho^{\chi,c}A_s^k\\
        B_s^kF_\rho^{\chi,c}  &=  F_\rho^{\chi,c}B_s^{k\ovl c}
    \end{align*}
\end{lemma}
\begin{proof}
	By direct calculation using equations~\eqref{eq:abribbon} and~\eqref{eqn:ABF}.
\end{proof}
As mentioned earlier the commutation relations of ribbon operators with the star and plaquette operators can be interpreted as the ribbon operators generating excitations at the ending sites of their respective ribbons, when applied to the ground state. The next lemma sheds more light on this interpretation.

\begin{lemma}\label{lem:triv_charge}
    Let $\rho$ be an open ribbon and $F_\rho^{\chi,c}$ an associated ribbon operator. Then, if $s=\partial_0\rho$ or $s=\partial_1\rho$, it holds
    \begin{align*}
        [F_\rho^{\chi,c},A_s] = 0&\iff \chi = \mr{id}.
    \end{align*}
    Similarly
    \begin{align*}
        [F_\rho^{\chi,c},B_s] = 0&\iff c = e.
    \end{align*}
\end{lemma}
\begin{proof}
    First note, that $F_\rho^{\chi,c}\neq 0$ as well as $A_s$ and $B_s$, since the ground state is not contained in the respective kernels. Note that the respective implications from the right hand side to the left hand side are true by Lemma~\ref{lem:ab_tool_3}.
    For the first statement we see, using Lemma~\ref{lem:ab_tool_3} and \cite[Theorem 27.15]{HeRo:1970},
    \begin{align*}
        A_sF_\rho^{\chi,c}A_s=\frac1{|G|}\sum_{k\in G}\chi(k)F_\rho^{\chi,c}A_s = \delta_{\chi,\mr{id}}F_\rho^{\chi,c}A_s.
    \end{align*}
	Thus, if $[A_s,F_\rho^{\chi,c}]=0$ we have $A_sF_\rho^{\chi,c}=\delta_{\chi,\mr{id}}F_\rho^{\chi,c}A_s$ which is only true, if $\chi = \mr{id}$.
    Using a similar derivation for the second statement we get
    \begin{align*}
        [F_\rho^{\chi,c},B_s] = 0    &\implies   \delta_{c,e}F_\rho^{\chi,c}B_s=F_\rho^{\chi,c}B_s
    \end{align*}
    and thus $c=e$.
\end{proof}
If $\rho$ is an open ribbon, then the excitations created at its ends by applying some ribbon operator on the ground state can be detected by certain local operators (see \cite[Section B 9.]{PhysRevB.78.115421} for detailed definitions).
A particularly useful example is that of certain projections (compare \cite[Section C 3.]{PhysRevB.78.115421}). For this let $s=\partial_0\rho$ the initial site of $\rho$. Then
    \begin{align}\label{eq:chapro}
        D_s^{\xi,d}:=\frac1{|G|}\sum_{k\in G}\ovl{\xi(k)}A_s^kB^d
    \end{align}
    detects the charge created by $F_\rho^{\chi,c}$ in the following sense:
    \begin{align*}
        D_s^{\xi,d}F_\rho^{\chi,c}\Omega   =  \frac1{|G|}\sum_{g,h\in G}\ovl{\xi(g)}\ovl{\chi(\ovl gh)}F_\rho^{\ovl c,h}A_s^{g}B_s^{\ovl cd}\Omega  = \delta_{\xi,\chi}\delta_{c,d}F_\rho^{\chi,c}\Omega.
    \end{align*}
    Here we used Lemma~\ref{lem:ab_tool_3} and $(\xi,d),(\chi,c)$ denote irreducible representations of $\mc D(G)$.
    Note, that by Lemma~\ref{lem:indeprib}, $D_s^{\xi,d}\Omega=\delta_{\xi,\mr{id}}\delta_{d,e}\Omega$ and in particular the projection onto the ground state is given by $D_s^{\mr{id},e}$. In subsequent sections we will use the notion $D_s:=D_s^{\mr{id},e}$.

Under some circumstances we can extend ribbons by triangles without changing an associated ribbon operator. This will be of some use later.
\begin{lemma}\label{lem:ab_tool_1}
    Let $\rho$ be an open ribbon and denote $s_0:=\partial_0\rho$ and $s_1:=\partial_1\rho$. Pick $c\in G$ and an irreducible representation $\chi$ of $G$. If there is a direct triangle $\tau$ such that $\tau\rho$ is a ribbon the following holds:
    \begin{align*}
        [F_\rho^{\chi,c},A_{s_0}]=0\implies F_\rho^{\chi,c} = F_{\tau\rho}^{\chi,c}
    \end{align*}
    The analogue statement holds true if $\rho\tau$ is a ribbon and the ribbon operator commutes with the star operator at $s_1$.

    If there is a dual triangle $\tau'$ such that $\tau'\rho$ is a ribbon then
    \begin{align*}
         [F_\rho^{\chi,c},B_{s_0}]=0\implies F_\rho^{\chi,c} = F_{\tau'\rho}^{\chi,c}
    \end{align*}
    and again an analogue statement holds true if $\rho\tau'$ is a ribbon.
\end{lemma}
\begin{proof}
    By Lemma~\ref{lem:triv_charge} $[F_\rho^{\chi,c},A_{s_0}]=0$ implies $\chi =\mr{id}$. Hence $F_\rho^{\chi,c}=F_\rho^{\mr{id},c}$ and therefore
    \begin{align*}
        F_{\tau\rho}^{\mr{id},c} = \sum_{g,k\in G}T_\tau^gF_\rho^{\ovl c,\ovl gk} = F_{\rho}^{\mr{id},c}
    \end{align*}
    since $\sum_{g\in G}T_\tau^g=I$.
    Analogously the other case.
    For if $\tau'\rho$ is a ribbon $ [F_\rho^{\chi,c},B_{s_0}]=0\implies c=e$ and
    \begin{align*}
        F_{\tau\rho}^{\chi,e} = \sum_{g,k\in G}\ovl{\chi(k)}L_{\tau'}^{e}\delta_{g,e}F_\rho^{e,\ovl gk}=F_\rho^{\chi,e}
    \end{align*}
    and again analogously for the second case.
\end{proof}
Since $G$ is abelian we also have that ribbon operators of closed ribbons commute with all star and plaquette operators.
\begin{lemma}
    Let $\rho$ be any closed ribbon. Then for all $h,g,k\in G$
    \begin{align*}
        [F_\rho^{h,g},A^k]=0=[F_\rho^{h,g},B^k].
    \end{align*}
\end{lemma}
The proof can be found in~\cite[Appendix B.5]{PhysRevB.78.115421}. A somewhat weaker statement of this is also true if we remove one triangle from a closed ribbon.
\begin{lemma}\label{lem:ab_tool_2}
    Let $\rho$ be an open ribbon such that there is a direct triangle $\tau$ with $\tau\rho$ is a closed ribbon. Then, with $\chi,c,s_0$ as above,
    \begin{align*}
        [A_{s_0},F_\rho^{\chi,c}]= 0\implies [B_{s_0},F_\rho^{\chi,c}]=0
    \end{align*}
    Given instead that there is a dual triangle $\tau'$ such that $\tau'\rho$ is a closed ribbon. Then
    \begin{align*}
         [B_{s_0},F_\rho^{\chi,c}]= 0\implies [A_{s_0},F_\rho^{\chi,c}]=0
    \end{align*}
\end{lemma}
\begin{proof}
    1.) The premises imply, by Lemma~\ref{lem:ab_tool_1}, that $F_{\rho}^{\chi,c}=F_{\tau\rho}^{\chi,c}$, and since $\tau\rho$ is a closed ribbon the claim follows.

    2.) The premises imply by Lemma~\ref{lem:ab_tool_1}, that $F_{\rho}^{\chi,c}=F_{\tau'\rho}^{\chi,c}$, and since $\tau'\rho$ is a closed ribbon the claim follows.
\end{proof}

\subsection{More on commutation relations}\label{sec:rib_comm}
In order to discuss statistics of superselection sectors later in Section~\ref{sec:irred_secs} we have to worry about commutation relations of ribbons.
In particular we want to know the commutation relations of ribbons which overlap at their ends and of ribbons that cross each other once, since then the commutation relations of associated ribbon operators reflect the braiding and fusion structure of irreducible representations of $\mc D(G)$ (c.f.~\cite{MR1951039}).

We start with some finite group $G$ and two ribbons $\rho,\sigma$. We say that $\rho,\sigma$ \emph{start at the same site} $s$ if there is a direct triangle $\tau$, a dual triangle $\tau'$, ribbons $\tilde\rho,\rho',\sigma'$ such that $\partial_0\tilde\rho=s$, $\rho'\cap\sigma'=\emptyset$ and $\rho=\rho'\tau'\tilde\rho,\sigma=\sigma'\tau\tilde\rho$ (in \cite{PhysRevB.78.115421} this is called a left joint).
The commutation relations of associated ribbon operators are then given by
\begin{align}\label{eq:rib_braid}
    F_\rho^{p,q}F_\sigma^{s,t} = F_\sigma^{ps\ovl p,pt}F_\rho^{p,q}.
\end{align}
Similarly, $\rho,\sigma$ \emph{end at the same site} $s$ if $\rho=\tilde\rho\tau'\rho'$, $\sigma=\tilde\rho\tau\sigma'$ and $\partial_1\tilde\rho=s$ (which is called a right joint in \cite{PhysRevB.78.115421}).
The corresponding commutation relations are given by
\begin{align*}
    F_\rho^{p,q}F_\sigma^{s,t}   &=  F_\sigma^{s,t\ovl q\ovl pq}F_\rho^{p,q}.
\end{align*}
Note that in the remaining possible cases for $\rho$ and $\sigma$, i.e. $\rho$ ends at the same site at which $\sigma$ starts we have that the ribbon operators commute.

We have particular interest in the commutation relations for finite abelian groups $G$. Here, for instance, equation~\eqref{eq:rib_braid} becomes
\begin{align*}
    F_\rho^{p,q}F_\sigma^{s,t} = F_\sigma^{s,pt}F_\rho^{p,q}.
\end{align*}
Furthermore, with the notation $(\chi,c),(\xi,d)$ for irreducible representations of $\mc D(G)$, the commutation relations for ribbons $\rho,\sigma$ starting at the same site give
\begin{align*}
    F_\rho^{\chi,c}F_\sigma^{\xi,d} &=  \ovl{\xi(c)}F_\sigma^{\xi,d}F_\rho^{\chi,c},
\end{align*}
and for ribbons ending at the same site we have
\begin{align*}
    F_\rho^{\chi,c}F_\sigma^{\xi,d} &=  \xi(c)F_\sigma^{\xi,d}F_\rho^{\chi,c}.
\end{align*}
Now consider two ribbons $\rho,\sigma$ that \emph{cross each other}, meaning there are ribbons $\rho_1,\rho_2,\sigma_1,\sigma_2$ such that $\rho=\rho_1\rho_2$, $\sigma=\sigma_1\sigma_2$, $\rho_i\cap\sigma_i\neq\emptyset, i=1,2$ and $\partial_1\rho_1=\partial_0\sigma_2$, $\partial_1\sigma_1=\partial_0\rho_2$. I.e. $\rho_1,\sigma_1$ end at the same site as $\sigma_2,\rho_2$ start at.
Such a situation is illustrated in Figure~\ref{fig:rib_cross}. The commutation relations then are
\begin{align}\label{eq:abel_braid}
    F_{\rho}^{p,q}F_{\sigma}^{s,t}=F_\sigma^{s,\ovl pt}F_\rho^{p,\ovl sq}.
\end{align}
Applied on ribbon operators labelled by irreducible representation of $\mc D(G)$ this gives
\begin{equation}
	\label{eq:irrepbraid}
	\begin{split}
    F_\rho^{\chi,c}F_\sigma^{\xi,d} &=  \sum_{g,h\in G} \ovl\chi(\ovl dg)\ovl\xi(\ovl ch)F_\sigma^{\ovl d,h}F_\rho^{\ovl c,g}\\
        &=  \chi(d)\xi(c)F_\sigma^{\xi,d}F_\rho^{\chi,c},
	\end{split}
\end{equation}
where we used $F_\rho^{\chi,c}=\sum_{g\in G}\ovl\chi(g)F_\rho^{\ovl c,g}$ and with the usual notation $(\chi,c)$ and $(\xi,d)$ for irreducible representations of $\mc D(G)$.
\begin{figure}
    \centering
    \includegraphics[width=0.35\linewidth]{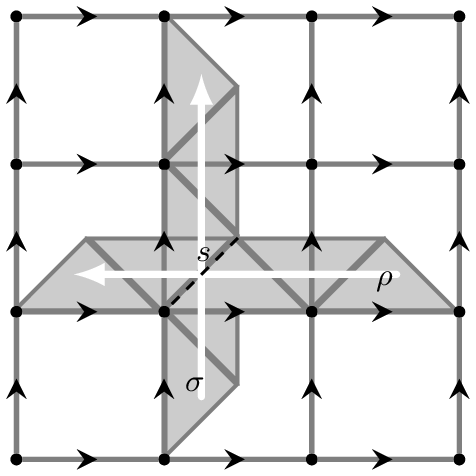}
    \caption{Two ribbons $\rho,\sigma$ crossing each other at site $s$. }
    \label{fig:rib_cross}
\end{figure}

With the commutation relations at hand we can prove the following technical lemma, which will be used later in one of the proofs for Haag duality.
It basically states that if at a site $s$ there is an excitation of the ground state created by multiple ribbons then it can be created by a single ribbon ending at $s$.
The excitations at the remaining spots different from $s$ are created by ribbon operators connecting those sites with each other.
\begin{lemma}\label{lem:tech}
    Let $\rho_1,\dots,\rho_n$ be open ribbons and $s$ be some site.
    Assume that $(\forall i\in\{1,\dots,n\})(\exists! j\in\{0,1\}):\ \partial_j\rho_i=s$. This gives a map $\{1,\dots,n\}\ni i\mapsto j_i\in\{0,1\}$.
    Furthermore assume that for all $i,i'\in\{1,\dots,n\}$ it holds that $\partial_{1-j_i}\rho_i\neq\partial_{1-j_{i'}}\rho_{i'}$.
    Let $\chi_i, i=1,\dots,n$ be irreducible representations of $G$ and elements $c_i\in G$, $i=1,\dots,n$.
    Set $\chi:=\chi_1\cdots\chi_n$ and $c:=c_1\cdots c_n$.

    Then there are ribbons $\sigma_1,\dots,\sigma_{n-1}$ with $\{\partial_0\sigma_k,\partial_1\sigma_k|k=1,\dots,n-1\}=\{\partial_{1-j_i}\rho_i|i=1,\dots,n\}$,
    a ribbon $\gamma$ with $\partial_0\gamma=s$ and $\partial_1\gamma=\partial_{1-j_i}\rho_i$ for some $i\in\{1,\dots,n\}$, and irreducible representations $\xi_1,\dots,\xi_n$ of $G$ and elements $d_1,\dots,d_{n-1}\in G$ such that
    \begin{align*}
        F_{\rho_1}^{\chi_1,c_1}\cdots F_{\rho_n}^{c_n}\Omega = zF_{\sigma_1}^{\xi_1,d_1}\cdots F_{\sigma_{n-1}}^{\xi_{n-1},d_{n-1}}F_\gamma^{\chi,c}\Omega
    \end{align*}
    where $z\in\mb C$ and $|z|=1$.
\end{lemma}
\begin{proof}
    The proof works by induction over the number of ribbons.
	By means of inversions of ribbons, i.e. Lemma~\ref{lem:ribinv}, we can assume without loss of generality that $j(\{1,\dots,n\})=\{0\}$ for any $n>0$. In other words we assume that all ribbons involved have their starting point at $s$ since otherwise we could invert them due to the aforementioned lemma.

    If $n=1$ the claim is trivial. We will elaborate on the case $n=2$ since this illustrates the basic idea of the proof.
    Let $\rho_1,\rho_2$ be ribbons as in the assumptions. Let $\chi_1,\chi_2$ be irreducible representations of $G$ and $c_1,c_2\in G$.
    Let $\overline\rho_1$ be an inversion of $\rho_1$ such that $\rho_2\ovl\rho_1$ is a ribbon.
    Then by Lemma~\ref{lem:indeprib} and Lemma~\ref{lem:ribinv} we have
    \begin{align*}
        F_{\rho_1}^{\chi_1,c_1}F_{\rho_2}^{\chi_2,c_2}\Omega &=  F_{\rho_1}^{\chi_1,c_1}F_{\rho_2}^{\chi_2,c_2}F_{\ovl\rho_1}^{\chi_2,c_2}F_{\rho_1}^{\chi_2,c_2}\Omega\\
        &=  zF_{\rho_1}^{\chi_1\chi_2,c_1c_2}F_{\rho_2}^{\chi_2,c_2}F_{\rho_2\ovl\rho_1}\Omega
    \end{align*}
    where $z$ is the factor given by the commutation relations in equation~\eqref{eq:rib_braid}.
    Now let $\sigma$ be a deformation of $\rho_2\ovl\rho_1$ such that $s\neq\sigma$. We then have
    \begin{align*}
        F_{\rho_1}^{\chi_1,c_1}F_{\rho_2}^{\chi_2,c_2}\Omega = zF_{\rho_1}^{\chi_1\chi_2,c_1c_2}F_\sigma^{\chi_2,c_2}\Omega
    \end{align*}
    as claimed.

    Now let $\rho_1,\dots,\rho_n$ be ribbons as in the preamble of the Lemma and assume that the claim holds for all any $n-1$ such ribbons.
    Let $\chi_1,\dots,\chi_n$ be irreducible representation of $G$ and $c_1,\dots,c_n\in G$. Set $\xi:=\chi_2\cdots\chi_n$ and $d:=c_2\cdots c_n$.
    Then
    \begin{align*}
        F_{\rho_1}^{\chi_1,c_1}\cdots F_{\rho_n}^{\chi_n,c_n}\Omega &=  zF_{\rho_1}^{\chi_1,c_1}F_{\sigma_2}^{\xi_2,d_2}\cdots F_{\sigma_{n-1}}^{\xi_{n-1},d_{n-1}}F_\gamma^{\xi,d}\Omega
    \end{align*}
    where the ribbons $\gamma,\sigma_k$, irreducible representations $\xi_k$ and $c_k\in G$ are corresponding to the claim.
    Let $\ovl\gamma$ be an inversion of $\gamma$ such that $\rho_1\ovl\gamma$ is a ribbon. Let $\sigma_1$ be a deformation of $\rho_1\ovl\gamma$.
    Again, using the same Lemmas as above we have
    \begin{align*}
        F_{\rho_1}^{\chi_1,c_1}\cdots F_{\rho_n}^{\chi_n,c_n}\Omega &=  \tilde zF_{\sigma_2}^{\xi_2,d_2}\cdots F_{\sigma_{n-1}}^{\xi_{n-1},d_{n-1}}F_{\rho_1}^{\chi_1,c_1}F_\gamma^{\xi,d}\Omega\\
            &= yF_{\sigma_2}^{\xi_2,d_2}\cdots F_{\sigma_{n-1}}^{\xi_{n-1},d_{n-1}}F_{\rho_1}^{\chi_1,c_1}F_\gamma^{\xi,d}F_{\ovl\gamma}^{\chi_1,c_1}F_{\gamma}^{\chi_1,c_1}\Omega\\
            &=  \tilde yF_{\sigma_2}^{\xi_2,d_2}\cdots F_{\sigma_{n-1}}^{\xi_{n-1},d_{n-1}}F_{\gamma}^{\xi\chi_1,cc_1}F_{\rho_1\ovl\gamma}^{\chi_1,c_1}\Omega\\
            &=  \tilde yF_{\sigma_2}^{\xi_2,d_2}\cdots F_{\sigma_{n-1}}^{\xi_{n-1},d_{n-1}}F_{\gamma}^{\xi\chi_1,cc_1}F_{\sigma_1}^{\chi_1,c_1}\Omega\\
            &=  \hat y F_{\sigma_2}^{\xi_2,d_2}\cdots F_{\sigma_{n-1}}^{\xi_{n-1},d_{n-1}}F_{\sigma_1}^{\chi_1,c_1}F_{\gamma}^{\xi\chi_1,cc_1}\Omega.
    \end{align*}
    The factors $\tilde z,y,\tilde y$ and $\hat y$ are products with $z$ and phase factors resulting from the commutation relations of the ribbon operators (c.f. Section~\ref{sec:rib_comm}).
    The last expression is of the form as in the claim.
\end{proof}
\section{Uniqueness of translational invariant ground state}
\label{sec:ground}
We now outline the proof that for each finite group $G$ (not necessarily abelian!), the quantum double model has a unique translational invariant ground state.\footnote{As was pointed out to us by Bruno Nachtergaele, the claim in~\cite{MR2345476} about uniqueness of the ground state is not entirely correct. By modifying finite volume boundary conditions and taking the thermodynamic limit, it is possible to obtain additional (algebraic) ground states. These states however are not invariant with respect to translations. Examples of such states are given by the ``single anyon states'', cf.~\cite[Prop 3.2]{toricendo}.} In case the model is defined by an oriented lattice on a compact surface it is known that the ground space degeneracy is the number of flat $G$-connections up to conjugation~\cite{MR1951039}, hence it is no surprise that in this infinite but topologically trivial setting we find a unique translational invariant ground state. The proof we discuss here is based on the proof in~\cite{phdnaaijkens}, where the full details can be found, which in turn is partly based on ideas of~\cite{MR2345476}.

Each term in the local Hamiltonians only acts on the bonds of a star or of a plaquette. Moreover, in the present situation of a square lattice, there is an obvious action of the group $\mathbb{Z}^2$ by translations. It follows that the local Hamiltonians $H_\Lambda$ are defined by a bounded, translation invariant interaction $\Phi$. Since the interaction is of bounded range and translationally invariant there is a corresponding one-parameter group $\alpha_t$ of automorphisms of $\alg{A}$ describing the time evolution~\cite{MR1441540}. The next lemma is useful when discussing ground states with respect to these dynamics.

\begin{lemma}[\hspace{-0.03cm}\cite{MR2345476}]
	\label{lem:state}
	Let $\omega$ be a state on a unital $C^*$-algebra $\alg{A}$, and suppose $X \in \alg{A}$ satisfies $X=X^*$, $X \leq I$, and $\omega(X) = 1$. Then $\omega(XY) = \omega(YX) = \omega(Y)$ for any $Y \in \alg{A}$.
\end{lemma}

The following characterisation of ground states for the quantum double model is inspired by results obtained in~\cite{MR2345476} for the special case of $G = \mathbb{Z}_2$.
\begin{proposition}\label{prop:kitground}
There exists a ground state $\omega_0$ for the dynamics of the quantum double model, which has the property that $\omega_0(A_s) = \omega_0(B_s) = 1$ for each site $s$. Moreover, every translation invariant ground state has this property.
\end{proposition}
\begin{proof}[Proof (sketch)]
	The star and plaquette algebras generate an abelian subalgebra of $\alg{A}$. We can identify each star and each plaquette with a classical Ising spin, hence this algebra describes two copies of the Ising model. The state we are looking for is the state with all spins in the up position. By the Hahn-Banach theorem there exists an extension to a state $\omega_0$ of $\alg{A}$. This is the state we are looking for: using Lemma~\ref{lem:state} it is straightforward to show that $-i \omega_0(X^* \delta(X)) \geq 0$ for all local observables $X$ (and $\delta$ the derivation implementing the dynamics). Hence $\omega_0$ is a ground state.

To show that any translational invariant ground state has this property, let $\omega_0$ be such a state. Since $A_s$ and $B_s$ are projections, it follows that $0 \leq \omega_0(A_s), \omega_0(B_s) \leq 1$. Because ground states minimise the mean energy $H_\Phi(\omega)$ by Theorem 6.2.58 of~\cite{MR1441540} one sees that we must have $\omega_0(A_s) = \omega_0(B_s) = 1$.
\end{proof}
To show that there is only one state on $\alg{A}$ with these properties, the idea is essentially to use Lemma~\ref{lem:state} again, just as it was used in the proof of the uniqueness of the translational invariant ground state of the toric code model~\cite{MR2345476}. The combinatorics, however, are much more involved. The proof consists of two steps. First we calculate the value of a ground state on certain products of projections acting on an individual site. In the second step this result is to calculate the expectation values of arbitrary local observables, showing that the ground state is completely fixed.

It was already remarked by Kitaev that the ground states of the quantum double model are related to so-called flat $G$-connections~\cite{MR1951039}. Here we have to consider \emph{local} observables, and hence it is enough to specify a $G$-connection for \emph{finite} parts of the system. The precise definition is a slight adaption from discrete gauge theory~\cite{MR2174961}:
\begin{definition}
	Let $F$ be a finite collection of faces and let $\Lambda \subset \bonds$ be the set of bonds bounding any face $f \in F$. A \emph{$G$-connection} $c$ is a map $c: \Lambda \to G$. A connection is called \emph{flat} if the monodromy around each face is trivial. That is, let $f \in F$ and list the edges $j_1, \dots j_n$ of $f$ in counter-clockwise order. Then the monodromy is trivial if $\sigma(c(j_1)) \sigma(c(j_2)) \cdots \sigma(c(j_n)) = e$, where $\sigma$ is as defined as follows: $\sigma(c(j)) = c(j)$ if the direction of $j$ coincides with the direction of the path around $f$, and $\inv{c(j)}$ otherwise. The set of all $G$-connections on $\Lambda$ will be denoted by $C_G(\Lambda)$, whose subset of flat connections is called $C^f_G(\Lambda)$.
\end{definition}
The constant map defined by $c_0(j) = e$ is trivially a flat $G$-connection, hence the set of flat connections is certainly non-empty. To each such a $G$-connection we can associate a projection, projecting on the basis vector $|c(j)\rangle$ at the site $j$. That is,
\[
	P_c = \prod_{j \in \Lambda} T^{c(j)}_{\tau(j)},
\]
where $\tau(j)$ is the direct triangle with edge $j$ whose orientation matches. Now, if $c$ is not a flat connection, there is a face $f$ with non-trivial monodromy. Let $s$ be a site with face $f$. Then it follows that $B_s P_c = 0$ (since $B_s$ projects on the subspace of trivial monodromy around the face $f$). With Lemma~\ref{lem:state} it follows that $\omega_0(P_c) = \omega_0(P_c B_s) = 0$ if $c$ is not flat. Now suppose that $c$ is a flat connection. Then one can show that $A^g_s P_c = P_{c'} A^g_s$, where $c'$ is also a flat connection. In fact by a sequence of such moves one can go from any flat connection $c$ to any other flat connection $c'$. With the same Lemma as before one then deduces the following  Lemma. For a detailed proof we refer to~\cite{phdnaaijkens}.
\begin{lemma}
Let $c \in C_G(\Lambda)$ and suppose that $\omega_0$ is a ground state for the quantum double model. Then $\omega_0(P_c) = 1/|C_G^f(\Lambda)|$ if $c$ is flat, and zero otherwise. Here $|C_G^f(\Lambda)|$ is the number of flat $G$-connections.
\end{lemma}

As remarked before the operators $L^g T^h$ acting on the same edge form a set of matrix units for the local algebra. Hence every local observable can be written as a linear combination of operators of the form $X = L P_c$, where $c$ is a connection and $L$ is a product of operators of the form $L^g_j$. By the argument above it follows that $\omega_0(X) = 0$ if $c$ is not flat. By systematically multiplying $X$ on the left (right) by star operators (plaquette operators), one can ``clean up'' the observable $X$, and show that $\omega_0(X)$ is either zero, or equal to  $\omega_0(P_{c'})$ for some flat connection $c'$. This argument leads to the following theorem~\cite{phdnaaijkens}:
\begin{theorem}
Kitaev's quantum double model on a square lattice on the plane has a unique translational invariant ground state $\omega_0$, completely determined by $\omega_0(A_s) = \omega_0(B_s) = 1$. This state is pure.
\end{theorem}
Purity follows because $\omega_0$ restricted to the abelian subalgebra generated by all $A_s$ and $B_p$ is multiplicative (hence pure). Since there exists a pure extension to $\alg{A}$ and by the argument above, the state $\omega_0$ is completely determined by the values on stars and plaquettes, it follows that $\omega_0$ must be pure. We will henceforth only consider this translational invariant ground state, and just refer to it as ``the'' ground state and will call the corresponding GNS representation the \emph{vacuum representation}.

If one inspects the full proof of the theorem given in~\cite{phdnaaijkens} carefully, one sees that in fact $\omega_0(AB) = \omega_0(A) \omega_0(B)$ if $A$ and $B$ are local, and their supports are sufficiently far removed from each other. This is related to the approximate split property, which will be discussed in Section~\ref{sec:split}.

\section{Haag duality}\label{sec:haagd}
The main result in this paper is the proof that Haag duality holds in the GNS representation of the translational invariant ground state for certain cone-like regions. We first introduce some definitions to make clear what we mean with a ``cone''. With these definitions we then discuss the proof. What is essential in our proof is a good understanding of how one can build up the Hilbert space of the ground state representation from excitations of the ground state. In particular, how one can obtain those excitations that are localised in a cone, by acting with the appropriate ribbon operators.
We use this to reduce the problem to a commutation problem of algebras acting on a smaller Hilbert space, consisting only of excitations \emph{inside} the cone. The ground state vector is cyclic for this Hilbert space, with respect to the cone algebra. This finally makes it possible to apply a result by Rieffel and Van Daele~\cite{RiDa:1975}, which relates the commutation properties of algebras to a density property of self-adjoint parts of algebras acting on a cyclic vector. In this way we circumvent the problem that the Reeh-Schlieder property (which says that the ground state vector is cyclic and separating for \emph{local} algebras) is not available, unlike for relativistic quantum field theories where it usually plays an important role in proving Haag duality~\cite{MR0438944,MR1147468}.

\subsection{Cones}
The main motivation to consider cone-like regions is given by the localisation regions of single excitations of the ground state.
These will turn out later to be suitably described by \emph{cones}.
How these cones are defined and which properties we need them to fulfill is described in the following.
We will state a list of requirements as a definition and then give a family of regions which fulfill this list.
Some of these requirements originate in the localisation properties of excitations sitting at the end of ribbons. Others are motivated as a technical requirement for proving a weaker form of the split property.
Most importantly cones should be ``ribbon connected'' in the sense that we can connect any site inside the cone with ribbons without leaving the cone. Furthermore it should be possible to translate any finite subset of the lattice into the cone using some lattice translation.

First we discuss what we mean by the boundary of a subset of $\Gamma$. We regard edges as a pair of vertices which are connected by an oriented bond. If we remove one of those vertices we also discard the bond.
\begin{definition}
    Let $\Lambda\subset\Gamma$ be a collection of edges and denote $\Lambda^c:=\Gamma\setminus\Lambda$.
    The \emph{interior} $\mr{int}(\Lambda^c)$ of $\Lambda^c$ is defined by the collection of edges in $\Lambda^c$ obtained by removing from $\mb Z^2$ all vertices contained in $\Lambda$ and discarding the associated bonds in $\Gamma$.
    The \emph{boundary} $\partial\Lambda^c$ of $\Lambda^c$ is then defined to be $\partial\Lambda^c:=\Gamma\setminus(\Lambda\cup\mr{int}(\Lambda^c))$ and we set $\partial\Lambda:=\partial\Lambda^c$.
\end{definition}
Note that the definition of $\partial\Lambda$ is symmetric under the exchange of $\Lambda$ and $\mr{int}(\Lambda^c)$.
Furthermore $\Lambda\cup\mr{int}(\Lambda^c)$ is a proper subset of $\Gamma$. That is to say $\partial\Lambda$ is the ``gap'' between $\Lambda$ and the interior of $\Lambda^c$.
\begin{definition}
    Given a subset $\Lambda\subset\Gamma$, a triangle $\tau\subset\Gamma$ and a ribbon $\rho\subset\Gamma$.
    We say that $\tau$ \emph{belongs to} or \emph{is contained in} $\Lambda$ if the edge of $\tau$ is in $\Lambda$.
    Similarly we say $\rho$ \emph{belongs to} $\Lambda$ if all triangles of $\rho$ belong to $\Lambda$.
    If this is the case we write $\tau\subset\Lambda$ and $\rho\subset\Lambda$.
\end{definition}
An illustration of this definition can be found in Figure~\ref{fig:triangle_region}.
\begin{figure}
    \centering
    \includegraphics[width=0.4\textwidth]{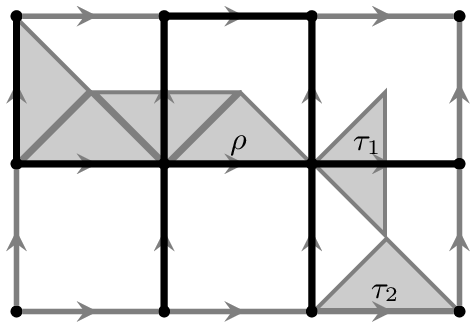}
    \caption{This image illustrates when triangles are contained in a region. The black lines indicate a collection $\Lambda$ of edges. The triangle $\tau_1$ and the ribbon $\rho$ are contained in $\Lambda$ whereas $\tau_2$ is not.}
    \label{fig:triangle_region}
\end{figure}
As we saw in Section~\ref{sec:proprib} excitations of the ground states are localised at sites and can be detected by star and plaquette operators.
Therefore, in order to distinguish whether an excitation is contained inside an area or not we have to specify when a site is, which is rather obvious.
Less clear on the other hand is the specification of a site sitting at the boundary of an area.
For our purposes and having an eye on Lemma~\ref{lem:ab_tool_1} we use the following notion.
\begin{definition}\label{def:bound_site}
    Let $\Lambda\subset\Gamma$ be again a subset and let $s=(v,f)$ be any site.
    Then $s$ is considered to be \emph{contained in} $\Lambda$, writing $s\in\Lambda$, whenever for any edge $e\in\Gamma$ with $\partial e=v$ it holds $e\in\Lambda$.

    We say that $s$ is \emph{contained in} $\partial\Lambda$, writing $s\in\partial\Lambda$ whenever $s\notin\Lambda$ and there are edges $e\in\Lambda$ and $e'\in\Lambda^c$ which bound $f$ or are neighbours of $v$.
\end{definition}
In other words $s=(v,f)\in\Lambda$ if the star at $v$ is contained in $\Lambda$, and $s\in\partial\Lambda$ if the star or the plaquette has non-empty intersection with $\Lambda$ and if $s\notin\Lambda$ (c.f. Figure~\ref{fig:cone_b}).

Unfortunately the definition of $s\in\partial\Lambda$ is not symmetric  under swapping the roles of $\Lambda$ and the \emph{interior} $\mr{int}(\Lambda^c)$ of $\Lambda^c$:
There might be sites that are contained in $\partial\Lambda$ that have empty intersection with $\Lambda^c$.
Nevertheless this definition is sufficient for our purposes since we just want to distinguish stars and plaquettes that are contained in $\mr{int}(\Lambda^c)$ from those having non-trivial intersection with $\Lambda$.
We will use this later on to move excitations that sit on the boundary of cones into the interior of the respective cone.
\begin{lemma}\label{lem:sites_outside}
    Let $\Lambda\subset\Gamma$ be some subset and let $s=(v,f)\in\mr{int}(\Lambda^c)$ be some site. Then for all edges $e$ ending at $v$ or bounding $f$ it holds $e\in\Lambda^c$.
\end{lemma}
\begin{proof}
    Assume that there was an edge $e\in\Lambda$ ending at $v$ or bounding $f$.
    Then in case it ends in $v$ we have $s\notin\mr{int}(\Lambda^c)$.
    In case $e$ bounds $f$ but does not end in $v$ we have that both $\partial_0e,\partial_1e\in\Lambda$. But then there is at least one edge $e'$ ending at $v$ and one of $\partial_0e,\partial_1e$ and hence $e'\in\partial\Lambda$. But then $s\notin\mr{int}(\Lambda^c)$.
\end{proof}
Finally the straightforward definition of a ribbon $\rho$ \emph{starting or ending at} $\partial\Lambda$ is given by requiring that the starting and ending sites $\partial_{0/1}\rho$ are contained in $\partial\Lambda$.

With this definition we have that a ribbon $\rho\subset\Lambda^c$ with, say, $\partial_0\rho\in\partial\Lambda$, is at most one triangle apart from $\Lambda$ in the following sense.
There is a ribbon $\rho_0\subset\Lambda^c$ with $\partial_0\rho_0\in\partial\Lambda$ such that $\rho_0\rho$ is a ribbon and $\rho_0$ is either a single triangle or a trivial ribbon.
(Here we have again Lemma~\ref{lem:ab_tool_1} in mind.). This situation is depicted in Figure~\ref{fig:cone_b}.

We now come to the definition of cones.
For any subset $\mc O\subset\Gamma$ and any point $y\in\mb Z^2$ we denote by $y+\mc O$ the subset in $\Gamma$ obtained by translating all pairs of vertices corresponding to edges in $\mc O$ by $y$.
\begin{definition}
    A subset $\Lambda\subset\Gamma$ is called \emph{cone} if it satisfies all of the following criteria.
    \begin{enumerate}
        \item\label{item:cone1}
            For any finite subset $\mc O\subset\Gamma$ there is a point $y\in\mb Z^2$ such that $y+\mc O\subset\Lambda$.
        \item
            For any pair of sites $s_0,s_1\in\Lambda$ there is a ribbon $\rho\subset\Lambda$ with $\partial_{0/1}\rho=s_{0/1}$.
        \item
            For any pair of sites $s_0,s_1\in\partial\Lambda$ there are ribbons $\rho_0,\rho_1\subset\Lambda^c$ and $\rho\subset\Lambda$ such that $\rho_0\rho\rho_1$ is a ribbon, $\partial_i\rho_i=s_i,i=0,1$ and $\rho_i, i=0,1$ are single triangles or trivial.
        \item
            For any pair of sites $s_0,s_1\in\partial\Lambda$ there is a ribbon $\rho\subset\Lambda^c$ such that $\partial_i\rho=s_i,i=0,1$.
    \end{enumerate}
\end{definition}
The first condition is of technical nature and plays a role when proving that the weak closure of cone algebras in the vacuum representation are factors of Type $II_\infty$ or of Type $III$ (see also Section~\ref{sec:split} and reference \cite{toricendo}).

The second and the third conditions express a kind of connectedness: Any pair of sites inside a cone $\Lambda$ can be connected with a ribbon, and sites at the boundary can be connected by ribbons that are contained in $\Lambda$ up to single triangles at the ends.
Both of them do not prohibit $\Lambda$ of having holes inside they just make sure that it is sufficiently connected in the aforementioned sense. The last condition ensures that that the complement $\Lambda^c$ is properly connected so that there are no holes in $\Lambda$.

As a result we can choose whether we want to connect sites at the boundary of the cone by ribbons that run in the exterior or in the interior of the cone up to triangles at the endpoints of the ribbon.
In particular for any ribbon $\rho\subset\Lambda^c$ with $\partial_{i}\rho\in\partial\Lambda, i=0,1$ there exist ribbons $\rho_0,\rho_1\subset\Gamma$ and $\tilde\rho\subset\Lambda$ such that $\rho_0\tilde\rho\rho_1$ is a ribbon, $\partial_{0}\rho_0=\partial_1\rho$, $\partial_1\rho_1=\partial_0\rho$ and $\rho_0,\rho_1$ are trivial ribbons or single triangles.
Furthermore, by condition \ref{item:cone1}, any cone is an infinite set.
\begin{figure}
    \centering
    \begin{subfigure}[t]{.48\textwidth}
        \centering
        \includegraphics[width=\textwidth]{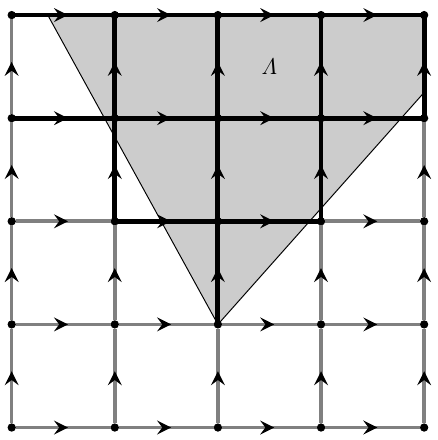}
        \caption{}
        \label{fig:cone_a}
    \end{subfigure}
    \hfill
    \begin{subfigure}[t]{.48\textwidth}
        \centering
        \includegraphics[width=\textwidth]{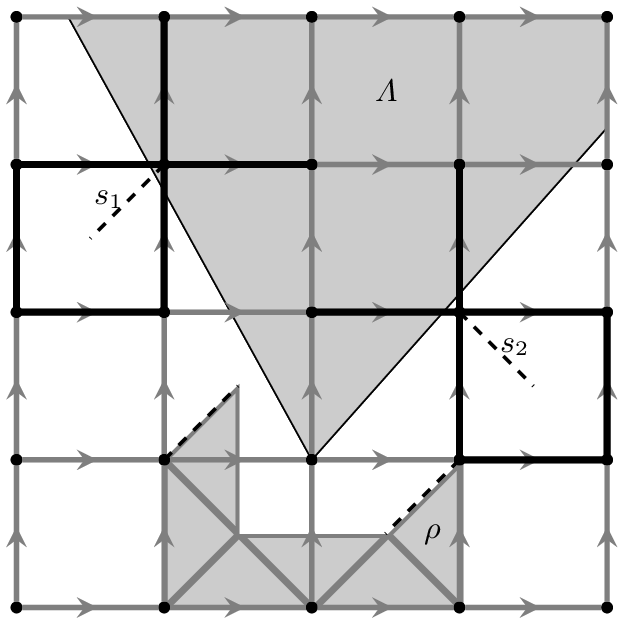}
        \caption{}
        \label{fig:cone_b}
    \end{subfigure}
        \caption{In both pictures the grey shaded region $\Lambda$ is a cone.
            (\subref{fig:cone_a}): Edges that are drawn black are either contained in $\Lambda$ or $\mr{int}(\Lambda^c)$ The grey bonds form $\partial\Lambda$.
            (\subref{fig:cone_b}):  Dotted lines indicate sites, especially $s_1\in\Lambda$ and $s_2\in\partial\Lambda$. The black lines highlight the edges belonging to the stars and plaquettes at $s_1$ and $s_2$. The ribbon $\rho$ connects a site at $\Lambda$ with a site in $\partial\Lambda$
        }
        \label{fig:cone}
\end{figure}
Examples for cones can be generated by those in $\mb R^2$: let $l_1\neq l_2$ be two semi-infinite lines in $\mb R^2$ emanating from a common point in $\mb Z^2$ and enclosing an angle smaller than $\pi$.
Denote by $\Lambda$ the set of edges that are contained in area enclosed by or have non-empty intersection with the two lines (see also Figure~\ref{fig:cone_a}).
It can be easily checked that $\Lambda$ is a cone.

In the following, and if not specified otherwise, $\Lambda$ will be a cone.

\subsection{Haag Duality}
In this section we prove Haag duality of cone algebras in the vacuum representation.
The proof is based on ideas developed in \cite{haagdtoric} and subdivides into several steps.
First we consider certain subspaces of the vacuum representation Hilbert space that are invariant under the action of cone algebras and show that the cone algebras are completely determined by this restriction.
Secondly we show that these subspaces are also invariant under the commutants of the algebras associated to the complement of the cones.
The last step consists in showing that linear combinations of the self-adjoint parts of the restricted cone algebras generate the above subspaces.
Using these facts together with a result by Rieffel and van Daele \cite{RiDa:1975} we can conclude Haag duality for the cone algebras.

For the start consider a cone $\Lambda\subset\Gamma$, denote the associated cone algebra by $\mc A(\Lambda)$ and by $\mc A(\Lambda^c)$ the one of the complement. The ground state's cyclic (GNS) representation is given by the tuple $(\pi_0,\Omega,\mc H_0)$ where the state itself is referred to as $\omega_0$. For any region $\mc O\subset\Gamma$ we denote the weak closure of $\mc A(\mc O)$ by $\mc R_{\mc O}:=\pi_0(\mc A(\mc O))''$.
As sketched above we aim at finding a subspace $\mc H_\Lambda\subset\mc H_0$ such that $\Omega$ is cyclic for $\mc R_\Lambda$. Again, we will identify operators $A\in\mc A$ with their image under $\pi_0$.

Let $\rho$ be a ribbon and let again $\mc F_\rho:=\{F^{h,g}_\rho\,|\, h,g\in G\}$ be the algebra linearly generated by all ribbon operators at $\rho$.
Note that the inclusion $\mc F_\rho\subseteq\bigotimes_{e\in\rho}\mc A_e$ is usually proper since $\mc F_\rho$ can be viewed as the subset of elements of the right hand side singled out by the commutation relations given by equation \eqref{eqn:ABF} (c.f.\cite[B.8]{PhysRevB.78.115421}).
For cones $\Lambda$ we denote by $\mc F_\Lambda:=\bigcup_{\rho\subset\Lambda}\mc F_\rho$ the algebra of ribbon operators localised in $\Lambda$. Analogously we denote $\mc F_{\Lambda^c}$ the algebra of ribbon operators localised in $\Lambda^c$.

The first observation is that products of operators in $\mc F_\Lambda$ and $\mc F_{\Lambda^c}$ generate a norm-dense subspace of $\mc H_0$ when applied to $\Omega$ (compare also \cite{haagdtoric}).
\begin{lemma}\label{lem:rib_dense}
    Given that $\Lambda\subset\Gamma$ is a cone we have with the notation from above:
    \begin{align*}
        &\overline{\mc F_\Lambda\mc F_{\Lambda^c}\Omega}^{\|\cdot\|}=\mc H_0.
    \end{align*}
\end{lemma}
\begin{proof}
    Single triangle operators are contained in $\mc F_\Lambda$ and $\mc F_{\Lambda^c}$. Since they form a basis of the edge algebras, operators in $\mc A_{\mr{loc}}(\Lambda)$ and $\mc A_{\mr{loc}}(\Lambda^c)$ are contained in $\mc F_\Lambda$ and $\mc F_{\Lambda^c}$, respectively.
    But those are norm-dense in $\mc A(\Lambda)$ and $\mc A(\Lambda^c)$, respectively, and together with cyclicity of $\Omega$ we arrive at the claim.
\end{proof}

\begin{definition}
    Let $\Lambda\subset\Gamma$ be a cone. We set $\mc H_\Lambda:=\overline{\mc F_\Lambda\Omega}^{\|\cdot\|}\subset\mc H_0$ and write $P_\Lambda$ for the projection onto $\mc H_\Lambda$.
\end{definition}
This subspace turns out to be left invariant by observables localised in the cone. Furthermore such observables are completely determined by their restriction to this space. The proof of this is the same as in \cite[Lemma 3.5]{haagdtoric} and we won't repeat it here.
\begin{lemma}\label{lem:invariant}
    For any cone $\Lambda\subset\Gamma$ the subspace $\mc H_\Lambda\subset\mc H_0$ is invariant under $\mc A(\Lambda)$, i.e. $\mc A(\Lambda)\mc H_\Lambda\subset\mc H_\Lambda$. Furthermore any element $A\in\mc R_\Lambda$ is completely determined by its restriction to $\mc H_\Lambda$.
\end{lemma}
As a consequence we have that $P_\Lambda\in\mc R_\Lambda'$.
One basic observation in the proof is that $\mc F_\Lambda$ is dense in $\mc A(\Lambda)$.
The next step consists of showing that a similar but less obvious statement holds true for the operators commuting with those localised in $\Lambda^c$.
The main idea is to show that we can characterise $\mc H_\Lambda^\perp$ by certain ribbon operators in $\mc F_{\Lambda^c}$ namely those which create non-trivial excitations in $\mr {int}(\Lambda^c)$.

Next we show that observables in the commutant of $\mc A(\Lambda^c)$ leave this vector space invariant. The basic idea is the same as that of the proof of \cite[Lemma 3.6]{haagdtoric}: We can characterize vectors of the form $F_1\cdots F_n\Omega$ to lie either in $\mc H_\Lambda$ or in $\mc H_{\Lambda}^\perp$ where $F_1,\ldots,F_n\in\mc A$ are ribbon operators.
Namely if $F_1\cdots F_n\Omega$ contains non-trivial excitations in $\mr{int}(\Lambda^c)$ then it is contained in $\mc H_\Lambda^\perp$. If there are no excitations in $\mr{int}(\Lambda^c)$ contained in this vector then it belongs to $\mc H_\Lambda$.
The next two lemmas show this in a stronger sense, namely that the orthogonal relation in the first case holds even if we apply any operator from $\mc A(\Lambda^c)'$ to the vector.
The idea is to detect excitations with star and plaquette operators acting on the ending sites of the corresponding ribbons. For this recall the definition of the projections $D_s^{\chi,c}$ in equation~\eqref{eq:chapro} acting at a site $s$.
To say that there is a charge in $\mr{int}(\Lambda^c)$ created by some ribbon operator amounts to seeing that there is some site $s\in\mr{int}(\Lambda^c)$ such that $D_s^{\mr{id},e}$ does not commute with this operator.
Note that this follows from the discussion in Section~\ref{sec:proprib}, especially the part around equation~\eqref{eq:chapro}.

The following three lemmas are essential in gaining a better understanding of the Hilbert space $\mc{H}_\Lambda$.
\begin{lemma}\label{lem:orth_charge_1}
    Let $\hat F:=F_1\cdots F_n\in\mc F_{\Lambda^c}$ be a product of ribbon operators associated to ribbons in $\Lambda^c$. Then the following holds:
    \begin{equation}
        \begin{split}
            &   \left(\exists s\in\mr{int}(\Lambda^c):[A_s,\hat F]\neq0\lor[B_s,\hat F]\neq 0\right)\\
                &\implies   \left((\forall F,C\in\mc F_\Lambda )(\forall X\in\mc A(\Lambda^c)'):(\hat FF\Omega,XC\Omega)=0\right).
        \end{split}\label{eq:cone_exc_1}
    \end{equation}
    Especially the left hand side implies $\hat F\Omega\in\mc H_\Lambda^\perp$.
\end{lemma}
\begin{proof}
    First note that because of Lemma~\ref{lem:sites_outside} $s\in\mr{int}(\Lambda^c)$, implies $A_s,B_s\in\mc F_{\Lambda^c}$.
    The proof works by repeated use of the lemmas of the discussion in Section~\ref{sec:proprib}.

    It is sufficient to work with ribbon operators labelled by irreducible representations of $\mc D(G)$ as defined in equation~\eqref{eq:abribbon}.
    Consider arbitrary such ribbon operators $\hat F_1,\dots,\hat F_n\in\mc F_{\Lambda^c}$ and let $C,F\in\mc F_\Lambda$ be some operators. By definition of $\mc F_\Lambda$ the operators $C$ and $F$ are sums of products of ribbon operators localised in $\Lambda$.
    For convenience we set $\eta := \hat F_1\cdots\hat F_nF\Omega\in\mc F_{\Lambda^c}\mc F_\Lambda\Omega$ and $\zeta := C\Omega\in\mc H_\Lambda$.

    Now for the proof of equation~\eqref{eq:cone_exc_1}, namely that if there are excitations in $\eta$ created by $\hat F_1,\dots,\hat F_n\in\mc F_{\Lambda^c}$ then $\eta$ is orthogonal to $X\zeta$ for all $C,F\in\mc F_\Lambda$ and $X\in\mc A(\Lambda^c)'$.

    Assume there exists a site $s\in\mr{int}(\Lambda^c)$ whose star operator $A_s$ does not commute with $\hat F_1\cdots\hat F_n$.
    Then, by Lemma~\ref{lem:ab_tool_3} and locality, we have
    \begin{align*}
        (\eta,X\zeta)   &=  \frac1{|G|}\sum_{k\in G}(\hat F_1\cdots\hat F_nFA_s^k\Omega,X\zeta)\\
            &=  \frac 1{|G|} \sum_{k\in G}\hat\chi_1(k)\cdots\hat\chi_n(k)(\eta,X\zeta)
    \end{align*}
    where $\hat\chi_j(k)$ either coincides with the corresponding term of the non-trivial representation of $\hat F_j$ if it doesn't commute with $A_s$, or $\hat\chi_j(k)=1$.
    Since for abelian groups the product of irreducible representations is again irreducible (they are all 1-dimensional), the right hand side equals $0$ since the appearing product representation is non-trivial. If the product representation was trivial then $[A_s,\hat F_1\cdots\hat F_n]=0$ and hence would contradict the assumptions (see Lemma~\ref{lem:triv_charge}). Thus we arrive at $(\eta,X\zeta) = 0$.

    Assume that there is a site $s\in\mr{int}(\Lambda^c)$ such that the associated plaquette operator $B_s$ does not commute with $\hat F_1\cdots\hat F_n$.
    Then there is at least one $j\in\{1,\dots,n\}$ with $[B_s,\hat F_{\rho_j}^{\chi,c}]\neq 0$ implying $c\neq e$ due to the commutation relations, see Lemma~\ref{lem:triv_charge}. More general there is a $k\in G$ with $k\neq e$ such that
    \begin{align*}
        \hat F_1\cdots\hat F_nB_s = B_s^k\hat F_1\cdots\hat F_n
    \end{align*}
    giving
    \begin{align*}
        (\eta,X\zeta)   =  (B_s^k\hat F_1\cdots\hat F_nC\Omega,\zeta)  =  (\eta,XFB_s^k\Omega)    =  0.
    \end{align*}
	This completes the proof.
\end{proof}
\begin{lemma}\label{lem:orth_charge_2}
    Let $\hat F:=F_1\cdots F_n\in\mc F_{\Lambda^c}$ be a product of ribbon operators associated to ribbons in $\Lambda^c$. Then the following holds:
     \begin{align}\label{eq:cone_exc_2}
         \left(\forall s\in\mr{int}(\Lambda^c):[A_s,\hat F]=0\land[B_s,\hat F]=0\right)\implies \hat F\Omega\in\mc H_\Lambda
    \end{align}
\end{lemma}
\begin{proof}
    Again, as in the previous proof, it is sufficient to work with ribbon operators labelled by irreducible representations of $\mc D(G)$.
    First some remarks about some general simplifications we are are going to assume.
    In case two ribbons $\rho,\sigma$ have the same starting and ending sites then, by Lemma~\ref{lem:indeprib}, one of them can be deformed into the other, giving
    \begin{align}\label{eqn:rib_manip}
        \hat F_\rho^{\chi,c}\hat F_\sigma^{\tau,d}\Omega = \hat F_\rho^{\chi\tau,cd}\Omega.
    \end{align}
    We always can assume that there are non-trivial and non-closed ribbons in the product $\hat F_1\cdots\hat F_n$. If ribbon operators associated to closed ribbons appeared then we simply could commute them past the other operators in $C$ to $\Omega$ where they leave $\Omega$ invariant.
    This can be seen by noting that if $\rho$ is a closed ribbon and $\rho=\rho_1\rho_2$ is a partition into ribbons then by Lemma~\ref{lem:ribinv} and Lemma~\ref{lem:indeprib} we have $\hat F_{\rho}^{\chi,c}\Omega=\hat F_{\rho_1}^{\chi,c}\hat F_{\ovl\rho_2}^{\ovl\chi,\ovl c}\Omega=\Omega$. Here $\ovl\rho_2$ is an inversion of $\rho_2$ which, by construction, starts and ends at the same sites as $\rho_1$.
    Due to the commutation relations of ribbons, see the discussion in Section~\ref{sec:rib_comm}, we may pick up some phase factors which will not be important here.

    In case that there are two ribbon operators $\hat F_1,\hat F_2$ associated to open ribbons $\rho_1,\rho_2$ such that $\rho_1\rho_2$ is a closed ribbon we can write them as a product of a ribbon operator of a closed ribbon and an operator associated to an open ribbon.
    To see this we move $\hat F_1$ and $\hat F_2$ to each other using the commutation relations of ribbons.
    Then we use Lemma~\ref{lem:abdecompose} and the remark after equation~\eqref{eq:abribbon} to find
    \begin{align}\label{eq:ribcompose}
        \hat F_{\rho_1}^{\chi,c}\hat F_{\rho_2}^{\xi,d} =  \hat F_{\rho_1}^{\chi,c}\hat F_{\rho_1}^{\ovl\xi,\ovl d}\hat F_{\rho_1}^{\xi,d}\hat F_{\rho_2}^{\xi,d} = \hat F_{\rho_1}^{\chi\ovl\xi,c\ovl d}\hat F_{\rho_1\rho_2}^{\xi,d}.
    \end{align}
    We also can always assume that ribbons just appear at most once in each product by the remark following equation~\eqref{eq:abribbon}.

    Now we turn to the claim of the lemma, equation~\eqref{eq:cone_exc_2}.
	We are performing an induction over the number of ribbon operators in $\hat F_1\cdots\hat F_n$, i.e. over the number of ribbon operators outside $\Lambda$.
    Let's start with $n=1$ and let $\hat F_1\in\mc F_{\Lambda^c}$ be a ribbon operator.
    Then we have that the ribbon $\rho\subset\Lambda^c$, to which $\hat F_1$ is associated to, is either of one of the following forms:
    It connects two sites in $\partial\Lambda$ or at least one ending site of $\rho$ is contained in $\mr{int}(\Lambda^c)$.

    Consider the case that $\rho$ connects two sites in $\partial\Lambda$.
    Taking a look at Definition~\ref{def:bound_site} we see that there are at most two triangles $\tau,\tilde\tau\subset\Lambda^c$ such that $\tau\rho\tilde\tau\subset\Lambda^c$ is a ribbon.
    By assumption and Lemma~\ref{lem:ab_tool_1} we have that $\hat F_1\Omega=\hat F_{\tau\rho\tilde\tau}\Omega$.
    But then we can invoke Lemma~\ref{lem:ribinv} and Lemma~\ref{lem:indeprib} to obtain a ribbon $\ovl\rho\subset\Lambda$ with $\hat F_{\tau\rho\tilde\tau}\Omega = \hat F_{\ovl\rho}\Omega$ and $\hat F_{\ovl\rho}\in\mc A(\Lambda)$.
    In case that $\rho$ has at least one ending site contained in $\mr{int}(\Lambda^c)$ Lemma~\ref{lem:triv_charge} (or an analogue calculation with equation~\eqref{eqn:ABF}) implies that $\hat F_1=I$.
    Hence in either case the vector is contained in $\mc H_\Lambda$.

    Now let $n>1$ be arbitrary but fixed and assume that equation~\eqref{eq:cone_exc_2} holds for all $\hat F_1,\dots,\hat F_{n-1}\in\mc A(\Lambda^c)$.
    Let therefore $\hat F_1,\dots,\hat F_{n}\in\mc A(\Lambda^c)$ be ribbon operators associated to ribbons in $\Lambda^c$ and set $\eta:=\hat F_1\cdots\hat F_n\Omega$.
    The remainder of the proof can be subdivided into different cases corresponding to the different configurations ribbons. We will relate some of them to each other and proof the remaining cases.
    The two main cases are the following:
    Firstly, there could be $k\leq n$ ribbons that start and end at $\partial\Lambda$.
    Secondly, there could be several ribbons having at least one end in $\mr{int}(\Lambda^c)$. See also Figure~\ref{fig:cases_first_lemma}.
    \begin{figure}
        \centering
        \includegraphics[width=0.4\textwidth]{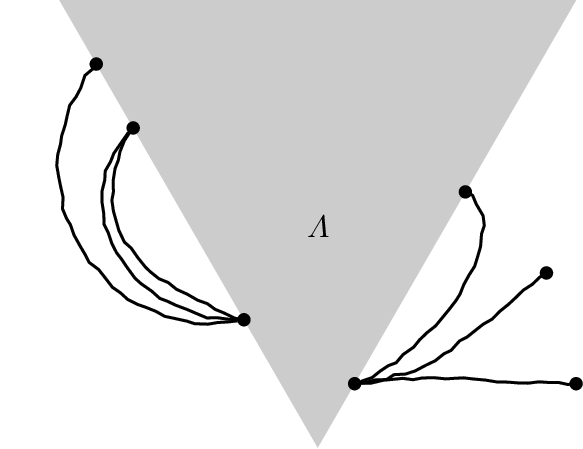}
        \caption{The two main cases in Lemma~\ref{lem:orth_charge_2} depicted in one image: On the left hand side of the cone $\Lambda$ is the case where only ribbons occur that connect sites of $\partial\Lambda$ whith each other. On the right hand side is the case with ribbons having ending sites in $\partial\Lambda$.}
        \label{fig:cases_first_lemma}
    \end{figure}

    The first main case can be handled as follows.
    Assume that there is a ribbon that connects two sites at $\partial\Lambda$, say $\rho_k$, $1\leq k\leq n$.
    Then we can commute the associated ribbon operator $\hat F_k$ in $\eta$ to the right in front of $\Omega$ thereby possibly obtaining a phase factor due to equation~\eqref{eq:irrepbraid}.
    But then, by using the argument from above, we can replace $\hat F_k$ with some operator $F_k\in\mc A(\Lambda)$ leaving a product of $n-1$ operators in $\mc A(\Lambda^c)$ in front of $F_k\Omega$.

    The second main case is a bit more involved.
    Consider that there is no such ribbon as in the first main case. If there is a ribbon $\rho$ having at least one ending site inside $\mr{int}(\Lambda^c)$ the following scenarios are possible.
    Firstly, one ending site of $\rho$ which is contained in $\mr{int}(\Lambda^c)$ does not coincide with an ending site of another ribbon occurring in $\eta$.
    Secondly, $\rho$ connects a site on $\partial\Lambda$ with a site in $\mr{int}(\Lambda^c)$ at which $k\geq 1$ other ribbons start or end.

    In the first case we find, by Lemma~\ref{lem:triv_charge}, that the associated ribbon operator $\hat F_\rho$ must be the identity operator.
    This reduces the product $\hat F_1\cdots\hat F_n\subset\mc A(\Lambda^c)$ in $\eta$ to a product of $n-1$ ribbon operators in $\mc A(\Lambda^c)$.

    In the second case we can assume that every of these $k$ ribbons connects to $\partial\Lambda$, since otherwise, we can just pick one of them that doesn't and use the previous procedure to remove it.
    Remember that we don't have to consider closed ribbons any more as well as open ribbons forming a closed loop.
    Now consider the ribbon operator $\hat F_\rho$ associated to $\rho$. We can safely assume that $\partial_1\rho$ is the site of interest. The other case can be treated in complete analogy.
    If there is a ribbon $\rho_l$ with $\partial_0\rho_l=\partial_1\rho$ then we first can deform $\rho$ into a ribbon $\tilde\rho$ such that $\tilde\rho\rho_l$ is a ribbon.
    On the level of ribbon operators this means first commuting the associated ribbon operator $\hat F_{\rho}$ in $\eta$ to the right in front of $\Omega$ and then using Lemma~\ref{lem:indeprib} to replace it with an operator $\hat F_{\tilde\rho}$.
    After that we use the commutation relations of ribbon operators again to move $\hat F_{\rho_1}$ to $\hat F_{\tilde\rho}$.
    We then can invoke equation~\eqref{eq:ribcompose} to obtain
    \begin{align*}
        \hat F_{\tilde\rho}^{\chi,c}\hat F_{\rho_l}^{\xi,d} =  \hat F_{\tilde\rho}^{\chi\ovl\xi,c\ovl d}F_{\tilde\rho\rho_l}^{\xi,d}.
    \end{align*}
    The ribbon $\tilde\rho\rho_l$ connects two sites at $\partial\Lambda$ and we can use a previous argument to replace $\hat F_{\tilde\rho\rho_l}$ in $\eta$ by a ribbon operator in $\mc A(\Lambda)$.

    If there is no ribbon $\rho_l$ with $\partial_0\rho_l=\partial_1\rho$ we pick one ribbon $\rho_l$ and apply Lemma~\ref{lem:ribinv} to replace it with a ribbon operator associated to a ribbon $\ovl\rho_l$ with $\partial_0\ovl\rho_l=\partial_1\rho$. But then we can proceed as before.
    Note that we also could have applied Lemma~\ref{lem:tech} instead to conclude the same for the second case.

    By induction we now can conclude that for any $n\in\mb N$ and any product of ribbon operators $\hat F_1,\cdots,\hat F_n\in\mc A(\Lambda^c)$ the relation in equation~\eqref{eq:cone_exc_2} holds true.
\end{proof}
\begin{lemma}\label{lem:inv}
    For any cone $\Lambda\subset\Gamma$ it holds $\mc A(\Lambda^c)'\mc H_\Lambda\subset\mc H_\Lambda$, hence $P_\Lambda\in\mc R_{\Lambda^c}$.
\end{lemma}
\begin{proof}
    Let $\hat F:=\hat F_1\cdots\hat F_n$ be a product of ribbon operators $\hat F_1,\ldots,\hat F_n\in\mc F_{\Lambda^c}$. Furthermore let $F,C\in\mc F_\Lambda$ and $X\in\mc A(\Lambda^c)'$ be any, non-zero, operators.
    For convenience set $\eta:=\hat FF\Omega$ and $\xi:=C\Omega$. Recall the definition of $D_s$ in equation~\eqref{eq:chapro}.

    By Lemma~\ref{lem:orth_charge_1} we have that if $(\eta,X\xi)\neq 0$ holds for all $F,C\in\mc F_\Lambda$ and $X\in\mc A(\Lambda^c)'$ then for any $s\in\mr{int}(\Lambda^c)$ the operator $\hat F$ commutes with $D_s$, i.e. $[F,D_s]=0$.
    Now by Lemma~\ref{lem:orth_charge_2} this implies $\eta\in\mc H_\Lambda$. To see this note that
    \begin{align*}
        \hat F\Omega\in\mc H_\Lambda\iff\left(\forall F\in\mc F_\Lambda:\hat FF\Omega\in\mc H_\Lambda \right)
    \end{align*}
    since $\hat FF=F\hat F$ and $\mc F_\Lambda\mc H_\Lambda\subseteq\mc H_\Lambda$. The other direction of this equivalence can be seen by assuming that the right hand side was true while the left hand was not which immediately leads to a contradiction since $I\in\mc F_\Lambda$. Summarizing this we obtain
    \begin{align}\label{eq:stuff}
        (\eta,X\xi)\neq 0\implies \eta\in\mc H_\Lambda
    \end{align}
    for all $\eta=\hat FF\Omega,\xi=C\Omega$ and $\hat F,F,C, X$ as above.

    By definition $\mc F_{\Lambda^c}$ contains all matrix units of the edge algebras $\mc A_e$ for $e\in\Lambda^c$ since the former are products of triangle operators. Hence products of ribbon operators $\hat F_1,\ldots,\hat F_n$ form a generating system of $\mc F_{\Lambda^c}$.
    Thus, by Lemma~\ref{lem:rib_dense}, the linear span of the set
    \begin{align*}
        \{\hat F_1\cdots\hat F_nF\Omega\,|\,\hat F_1,\dots\hat F_n\in\mc F_{\Lambda^c}\textnormal{ ribbon operators },F\in\mc F_\Lambda, n\in\mb N\}
    \end{align*}
    is a dense subspace of $\mc H$.
    From this we conclude that equation~\eqref{eq:stuff} holds for any $\eta\in\mc H$ and $\xi\in\mc H_\Lambda$. Therefore
    \begin{align*}
        (\forall\psi\in\mc H):\psi\in\mc H_\Lambda^\perp\implies\left((\forall\phi\in\mc H_\Lambda)(\forall X\in\mc A(\Lambda^c)'):(\psi,X\phi)=0\right)
    \end{align*}
    and we arrive at $\mc A(\Lambda^c)'\mc H_\Lambda\perp\mc H_\Lambda^\perp$.
\end{proof}
As the next step we want to consider the restrictions of the von~Neumann algebras $\mc R_\Lambda$ and $\mc R_{\Lambda^c}$ to $\mc H_\Lambda$.
By \cite[Proposition II.3.10]{Takesaki:2002} both restrictions are again von~Neumann algebras.
\begin{definition}
    For any cone $\Lambda\subset\Gamma$ we write $\mc A_\Lambda:=P_\Lambda\mc R_\Lambda P_\Lambda\upharpoonright_{\mc H_\Lambda}$ and $\mc B_\Lambda=P_\Lambda\mc R_{\Lambda^c}P_\Lambda\upharpoonright_{\mc H_\Lambda}$ as subalgebras of $\mc B(\mc H_\Lambda)$.
\end{definition}
By using similar techniques as in the proof of the lemmas~\ref{lem:orth_charge_1} and \ref{lem:orth_charge_2} we show that elements of the form $A_s+iB_s$ with $A_s\in\mc A_s$ and $B_s\in\mc B_s$ already generate $\mc H_\Lambda$ when applied on the ground state vector. Here $\mc{A}_s$ is the self-adjoint part of $\mc{A}_\Lambda$, and similarly for $\mc{B}_s$.
\begin{lemma}\label{lem:selfadjoint}
    Let $\mc A_s$ be the self-adjoint part of $\mc A_\Lambda$ and $\mc B_s$ that of $\mc B_\Lambda$. Then the set
    \begin{align*}
        \mc A_s\Omega+i\mc B_s\Omega
    \end{align*}
    is dense in $\mc H_\Lambda$.
\end{lemma}
\begin{proof}
    First note that since both $\mc A_s$ and $\mc B_s$ are real vector spaces it suffices to show for operators $F\in\mc F_\Lambda$ that $F\Omega$ and $iF\Omega$ are contained in $\mc A_s\Omega+i\mc B_s\Omega$.
    In order to do so we first show this to hold if $F$ is a finite product of ribbon operators in $\mc F_\Lambda$ and then conclude for general operators $F\in\mc F_\Lambda$ by a density argument.
    Essential now is the structure of the vector space $\mc H_\Lambda$ that we elaborated on earlier in Lemma~\ref{lem:rib_dense} and in the proofs of Lemma~\ref{lem:orth_charge_1} and \ref{lem:orth_charge_2}.
	This is to say that finite products of ribbon operators in $\mc F_\Lambda$ applied to the vacuum vector $\Omega$ sufficiently describe $\mc H_\Lambda$ and certain ribbon operators in $\mc F_{\Lambda^c}$ map $\Omega$ to vectors in $\mc H_\Lambda$ and can be expressed as the images of $\Omega$ of certain elements of $\mc F_\Lambda$.

    Throughout the proof we consider ribbon operators labelled by irreducible representations of the quantum double model and we can assume that the label is nontrivial for if it was trivial we just obtain the identity operator.
    Again we will use the charge projections $D_s^{\chi,c}$ introduced in equation~\eqref{eq:chapro} which project onto the excitation given by $(\chi,c)$ at site $s$. Especially recall that $D_s=D_s^{\mr{id},e}=A_sB_s$.
    Now let $F_1,\dots,F_n\in\mc F_\Lambda$ be ribbon operators with $n>0$ and set $F:=F_1\cdots F_n$.
    The idea is to construct self-adjoint elements of $\mc A_s$ and $\mc B_s$ by taking linear combinations of products of projections $A_s, B_s$ and products of ribbon operators in $\mc F_\Lambda$ and $\mc F_{\Lambda^c}$.
    These self-adjoint operators are chosen in such a way that they map the state vector to the same vector as $F$.
    Again, as in previous proofs, we will work with an induction over the number of ribbon operators in $F$.
    With the same argument as in the proof of Lemma~\ref{lem:orth_charge_2} we can assume that in $F$ there are no ribbon operators associated to closed ribbons or trivial ribbons.

    Let $n=1$ and let $\rho$ denote the corresponding ribbon. In case that both the star and the plaquette at least one of the ending sites of $\rho$, denoted by $s$, are contained in $\Lambda$ we set
    \begin{align*}
        \tilde F&:=FD_s+D_sF^*  &    &   \textnormal{and}    &      \hat F&:= i(FD_s-D_sF^*).
    \end{align*}
    Obviously these operators are selfadjoint hence contained in $\mc A_s$ and it can easily be checked that $\tilde F\Omega=F\Omega$ and $\hat F\Omega = iF\Omega$. Therefore $F\Omega$ and $iF\Omega$ belong to $\mc A_s\Omega$.

    Assume that at both ends of $\rho$ are contained in $\Lambda$ but the plaquettes at both sites are not contained in $\Lambda$. Then the stars are still contained in $\Lambda$, by definition (c.f. Definition~\ref{def:bound_site} and the discussion after) and the star operators are elements of $\mc A_s$.
    In case $[F,A_s]\neq 0$, with $s=\partial_0\rho$ or $s=\partial_1\rho$, it suffices to take
    \begin{align*}
        \tilde F&:=F A_s+A_sF^*    &    &   \textnormal{and}       &    \hat F&:=i(F A_s-A_sF^*),
    \end{align*}
    since then $\tilde F\Omega = F\Omega + \delta_{\chi,\mr{id}}\Omega=F\Omega$ an analogously $\hat F\Omega = iF\Omega$ where $\chi$ is part of the label of $F$.
    These operators are selfadjoint and $\tilde F,\hat F\in\mc A_s$ hence $F\Omega,iF\Omega\in\mc A_s\Omega$.

    If, however, $[F,A_s]=0$ we can use Lemma~\ref{lem:ab_tool_1} to extend $\rho$ with triangles $\tau,\tilde\tau$ such that $\tilde\rho:=\tau\rho\tilde\tau$ is a ribbon, and $\partial_0\tilde\rho,\partial_1\tilde\rho\in\partial\Lambda$.
    Furthermore we then have $F_{\tilde\rho}\Omega=F\Omega$.
    But now we can invoke Lemma~\ref{lem:ribinv} and Lemma~\ref{lem:indeprib} to find a ribbon $\ovl\rho\subset\Lambda^c$ such that $F\Omega=F_{\ovl\rho}^*\Omega$.
    Now we can set
    \begin{align*}
        \tilde F:=\frac12(F+F^*)+i\left(\frac i2(F_{\ovl\rho}-F_{\ovl\rho}^*)\right)
    \end{align*}
    and it can easily be checked that the ``real part'' of $\tilde F$ is an element of $\mc A_s$ and the ``imaginary part'' one of $\mc B_s$, hence $\tilde F\in\mc A_s+i\mc B_s$. By construction $F\Omega=\tilde F\Omega\in\mc A_s\Omega+i\mc B_s\Omega$. Similarly
    \begin{align*}
        \hat F:=\frac i2(F-F^*)+\frac i2(F_{\ovl\rho}+F_{\ovl\rho}^*)
    \end{align*}
    and $iF\Omega=\hat F\Omega\in\mc A_s\Omega+i\mc B_\Omega$.

    We now proceed by induction. Let $n>0$ be arbitrary but fixed and assume that the assertion holds for all $F_1,\dots,F_{n-1}\in\mc F_\Lambda$. Let $F_1,\dots,F_n\in\mc F_\Lambda$ be any non-trivial ribbon operators. If one of them was trivial then we could remove it and obtained $n-1$ factors.
    Again we have different cases to treat. First of all we handle the case where we can remove or combine ribbon operators leaving us with $n-1$ factors in the product. More precisely, consider that there are two ribbon operators associated to ribbons $\rho_i,\rho_k$ with $1\leq i,k\leq n$ such that they start and end at the same site.
    Then either $\partial_j\rho_i=\partial_j\rho_k, j=0,1$ or $\rho_i\rho_k$ is a closed ribbon.
    In either case in the product $F_1\cdots F_n$ we can bring $F_i$ and $F_k$ to the right by using the commutation relations of ribbon operators.
    Then we can use equation~\eqref{eqn:rib_manip} and the remark after equation~\eqref{eq:abribbon} to replace $F_iF_j$ in front of $\Omega$ with a single ribbon operator.
    If $\rho_i\rho_k$ is closed then we have $F_i^{\chi,c}F_j^{\xi,d}\Omega=F_{\rho_i}^{\chi\ovl\xi,c\ovl d}\Omega$.
    In case  $\partial_j\rho_i=\partial_j\rho_k, j=0,1$ we have $F_i^{\chi,c}F_j^{\xi,d}\Omega=F_i^{\chi\xi,cd}\Omega$.
    Again $(\chi,c)$ and $(\xi,d)$ are irreducible representations of $\mc D(G)$.
    That is, in both cases we end up with a product of $n-1$ ribbon operators in front of $\Omega$.
    This allows us to assume in the rest of the proof that in $F_1\cdots F_n$ each ribbon involved there is appearing exactly once.

    The rest of the proof can be divided into three main cases. Let again $F_1\cdots F_n$ be the product of non-trivial ribbon operators in $\mc F_\Lambda$. Assume that there are no such ribbons as in the previous case.
    Then there are three possibilities: either there exists a ribbon $\rho$ involved in the product such that $D_{\partial_i\rho}\in\mc F_\Lambda$ for at least one $i=0,1$, or all ribbons end at $\partial\Lambda$, or neither of both, i.e. $\partial_i\rho\notin\partial\Lambda$ and $D_{\partial_i\rho}\notin\mc F_\Lambda$.

    Consider the first main case, namely that $D_{\partial_i\rho}\in\mc F_\Lambda$ for $i=0$ or $i=1$ for at least one ribbon involved in $F_1\cdots F_n$. We set $s:=\partial_0\rho$ and without loss of generality we can assume that $F_\rho=F_n$ and $i=0$. If the ribbon operator was not $F_n$ we could use the commutation relations of ribbon operators to move this operator to the last place in the product.
    We can divide the treatment of this case into two different cases. The first case is that there is a site $s\in\Lambda$ such that $[F_1\cdots F_n,D_s]\neq 0$. In the other case we have that for all sites $s'\in\Lambda$ with $D_{s'}\in\mc F_\Lambda$ it holds $[F_1\cdots F_n,D_{s'}]=0$.

    Now for the first subcase of the first main case. If there is a site $s\in\Lambda$ with $[F_1\cdots F_n,D_s]\neq 0$ we can set
    \begin{align*}
        \tilde F&:=F_1\cdots F_nD_s+D_sF_1^*\cdots F_n^*    &&\textnormal{and}&       \hat F&:=iF_1\cdots F_nD_s-iD_sF_1^*\cdots F_n^*.
    \end{align*}
    Then $\tilde F, \hat F\in\mc A_s$ and it holds $F_1\cdots F_n\Omega=\tilde F\Omega$ and similarly $iF_1\cdots F_n\Omega=\hat F\Omega$.

    The case that for all sites $s'\in\Lambda$ with $D_{s'}\in\mc F_\Lambda$ it holds $[F_1\cdots F_n,D_{s'}]=0$ can be treated as follows.
    Since we assumed that there is at least one ribbon $\rho$ involved in the product, the corresponding ribbon operator is either trivial, by Lemma~\ref{lem:triv_charge}, or there is at least one additional ribbon ending or starting at one of the endpoints of $\rho$. We excluded the first case by assumption so we have to treat the second one.
    Therefore consider the situation where there are $k$ ribbons $\rho_{n-k},\dots,\rho_n$ in $F_1\cdots F_n$ ending at $s$.
	By Lemma~\ref{lem:triv_charge} the condition that the operators commute with the charge projector is equivalent to $\chi_{n-k}\cdots\chi_n=\mr{id}$ and $c_{n-k}\cdots c_n=e$ where $\chi_i$ are irreducible representations of $G$ and $c_i\in G$ with $i=n-k,\dots,n$.
    But by Lemma~\ref{lem:tech} we have that there are ribbons $\sigma_{n-k},\dots,\sigma_{n-1}$ such that they do not cross the site $s$, a ribbon $\gamma$ having $s$ as an ending site, irreducible representations $\xi_{n-k},\dots,\xi_{n-1}$ of $G$ and elements $d_{n-k},\dots,d_{n-1}\in G$ such that
    \begin{align*}
        F_{\rho_1}^{\chi_1,c_1}\cdots F_{\rho_n}^{\chi_n,c_n}\Omega &=  zF_{\rho_1}^{\chi_1,c_1}\cdots F_{\rho_{n-k-1}}^{\chi_{n-k-1},c_{n-k-1}}F_{\sigma_{n-k}}^{\xi_{n-k},d_{n-k}}\cdots F_{\sigma_{n-1}}^{\xi_{n-1},d_{n-1}}F_\gamma^{\chi,c}\Omega
    \end{align*}
    where $z\in\mb C,|z|=1$ and $\chi=\chi_{n-k}\cdots\chi_{n}$ and $c=c_{n-k}\cdots c_n$.
    The commutation relation with the charge projection now tells us that $\xi=\mr{id}$ and $c=e$, hence $F_\gamma^{\chi,c}=I$. This gives an expression with $n-1$ ribbon operators acting on $\Omega$ and we are done for this case.

    Let's turn to the second main case where in the product $F_1\cdots F_n\in\mc F_\Lambda$ there are only ribbons $\rho_i,i\in\{1,\dots,n\}$ involved whose ending sites are contained in $\partial\Lambda$.
    By definition, c.f. Definition~\ref{def:bound_site}, it holds for all $i\in\{1,\dots,n\}$ that $D_{\partial_k\rho_i}\neq\mc F_\Lambda, k=0,1$ so we cannot treat this in the manner as the first main case.
    In the proof of Lemma~\ref{lem:orth_charge_2} we used that we can replace ribbon operators associated to ribbons, which are contained in $\Lambda^c$ and which connect sites on $\partial\Lambda$, to ribbon operators of ribbons which are contained in $\Lambda$ and which connect the same sites, without changing the image of $\Omega$ under these operators.
    Of course, this works the other way round, too. So choosing
    \begin{align*}
        \tilde F:=\frac12\left(F_{\rho_1}\cdots F_{\rho_n}+F_{\rho_n}^*\cdots F_{\rho_1}^*\right)+i\left(\frac i2(F_{\tilde\rho_1}\cdots F_{\tilde\rho_n}-F_{\tilde\rho_n}^*\cdots F_{\tilde\rho_1}^*)\right)
    \end{align*}
    and
    \begin{align*}
        \hat F:=\frac i2\left(F_{\rho_1}\cdots F_{\rho_n}-F_{\rho_n}^*\cdots F_{\rho_1}^*\right)+\frac i2\left(F_{\tilde\rho_1}\cdots F_{\tilde\rho_n}+F_{\tilde\rho_n}^*\cdots F_{\tilde\rho_1}^*\right)
    \end{align*}
    will do the job.
    We used the notation $F_{\rho_i}$ instead of $F_i,i=1,\dots,n$ to indicate the dependence on the ribbon.
    As above $\tilde\rho_i$ indicates the ribbon obtained by extending $\rho_i$ by triangles corresponding to Lemma~\ref{lem:ab_tool_1} if necessary, and inverting it using Lemma~\ref{lem:ribinv}.
    Then $\tilde F,\hat F\in\mc A_s+i\mc B_s$ and it can be easily be verified that $\tilde F\Omega=F_1\cdots F_n\Omega$ and $\hat F\Omega=iF_1\cdots F_n\Omega$.

    It remains to treat the third main case. Consider there is no $\rho$ involved in $\hat F$ such that it falls under the two previous main cases. I.e. for any $\rho\subset\Lambda$ appearing in $\hat F$ at least one of the ending sites $s_i:=\partial_i\rho, i=0,1$ is such that $\mc D_{s_i}\notin\mc F_\Lambda$ and $s_i\notin\partial\Lambda$.
    Then, by construction of $\Lambda$ and by the Definition~\ref{def:bound_site}, the star operators at $s_i$ are still contained in $\Lambda$, i.e. $A_{s_i}\in\mc F_\Lambda$. Furthermore for each such $s_i$ there are triangles $\tau_i\in\Lambda$ such that $\tau_i\rho$ or $\rho\tau_i$ is a ribbon and $\partial_{i}\tau_i\in\partial\Lambda$.
    There are two cases appearing here: $[A_{s_i},F]=0$ for any such $s_i$ or $[A_{s_i},F]\neq 0$ for some $s_i$. In case $[A_{s_i},F]\neq 0$ for any $s_i$ we simply set
    \begin{align*}
        \tilde F&:=F_1\cdots F_nA_{s_i}+A_{s_i}F_1^*\cdots F_n^*    &&\textnormal{and}&       \hat F&:=iF_1\cdots F_nA_{s_i}-iA_{s_i}F_1^*\cdots F_n^*.
    \end{align*}
    Then $\tilde F, \hat F\in\mc A_s$ and it holds $F_1\cdots F_n\Omega=\tilde F\Omega$ and similarly $iF_1\cdots F_n\Omega=\hat F\Omega$.
    In case there is a $s_i$ such that $[A_{s_i},F]=0$ we first deform or invert any ribbon $\sigma$ involved in $F$, using Lemma~\ref{lem:indeprib} and \ref{lem:ribinv}, such that any of them has $s_i$ as final site, i.e. $\partial_0\sigma=s_i$ and any of them stays in $\Lambda$.
    This gives an expression $F\Omega=F'\Omega$ where $F'$ is again a product of ribbon operators in $\mc F_\Lambda$ together with a possible phase factor from the commutation relations. More importantly, $[F',A_{s_i}]=0$. Let $\sigma'$ denote these possibly deformed or inverted ribbons.
    Then, by Lemma~\ref{lem:ab_tool_1} there is a triangle $\tau\in\Lambda$ such that $\partial_1\tau=s_i$ and for any ribbon $\sigma'$ it holds $\tau\sigma'\subset\Lambda$ is a ribbon. Furthermore $\partial_0\tau\in\partial\Lambda$.
    If we apply this procedure to any of the ending sites $s_i$ of ribbons in $F$ for which $[F,A_{s_i}]=0$ we end up at the situation in the second main case from where we can proceed accordingly.

    This completes the third main case and also the proof of the claim.
\end{proof}
With these preparations we are finally in a position to prove the main theorem. In particular, the last lemma allows us to use the result of Rieffel and Van Daele mentioned before.
\begin{theorem}\label{thm:haagdual}
    Cone algebras of the quantum double model for finite abelian groups on the infinite square lattice satisfy Haag duality in the vacuum representation.

    More precisely, if $\Lambda\subset\Gamma$ is a cone then
    \begin{align*}
        \pi_0\left(\mc A(\Lambda^c)\right)'=\pi_0\left(\mc A(\Lambda)\right)''.
    \end{align*}
\end{theorem}
\begin{proof}
    The argument is exactly the same as that given in reference \cite{haagdtoric}. For the convenience of the reader, we will restate it here.

    It remains to prove $\mc A(\Lambda^c)'\subset\mc A(\Lambda)''$ since, by locality, the other direction already holds.
    By construction it holds that $\mc A_\Lambda\subset\mc B_\Lambda'$ (as sub-algebras of $\mc B(\mc H_\Lambda)$) and both, $\mc A_\Lambda$ and $\mc B_\Lambda$, are von-Neumann algebras on the same Hilbert space $\mc H_\Lambda$. Hence, by \cite[Theorem 2]{RiDa:1975}, the statement of Lemma~\ref{lem:selfadjoint} is equivalent to $\mc A_\Lambda=\mc B_\Lambda'$.

    Furthermore, by \cite[Proposition II.3.10]{Takesaki:2002}, it holds that $\mc B_\Lambda'=P_\Lambda\mc R_{\Lambda^c}'P_\Lambda\upharpoonright_{\mc H_\Lambda}$.
    Now let $B\in\mc R_{\Lambda^c}'$ and denote $B_\Lambda:=P_\Lambda BP_\Lambda\upharpoonright_{\mc H_\Lambda}\in\mc B_\Lambda'$. Then $B_\Lambda\in\mc A_\Lambda$ and, by Lemma~\ref{lem:invariant}, there exists a unique element $A\in\mc R_\Lambda$ such that $B_\Lambda=P_\Lambda AP_\Lambda\upharpoonright_{\mc H_\Lambda}$.

    To proof the claim it suffices to show that $B=A$. Pick any $\hat F\in\mc F_{\Lambda^c}$ and $F\in\mc F_\Lambda$. Then
    \begin{align*}
        B\hat FF\Omega = \hat F BF\Omega = \hat FB_\Lambda F\Omega = \hat FAF\Omega = A\hat FF\Omega
    \end{align*}
    giving $A=B$, by Lemma~\ref{lem:rib_dense}, and consequently $B\in\mc R_\Lambda$.
\end{proof}

\section{The approximate split property}\label{sec:split}
One can ask the question if the observable (von Neumann) algebra actually is isomorphic to $\mc{R}_\Lambda \otimes \mc{R}_{\Lambda^c}$ if $\Lambda$ is a cone, so that we can see the cone part and the outside as two separate, independent systems without any correlations between them. This turns out not to be the case, because $\mc{R}_{\Lambda}$ is not a factor of Type I (remark that if this was the case Haag duality would follow readily). The proof that these factors are not of Type I given in~\cite[Thm 5.1]{toricendo} works for general finite groups $G$. Nevertheless, a slightly weaker condition \emph{is} true. If we separate the cone $\Lambda$ from the complement of a slightly bigger cone $\Lambda'$, the resulting von Neumann algebra \emph{is} a tensor product of the observable algebras in the two disjoint regions. This follows from the approximate split property For the convenience of the reader we first recall the precise definition.\footnote{We note again that in previous work we called this the \emph{distal} split property.}
\begin{definition}
We say that $\pi_0$ has the \emph{approximate split property} if for each pair $\Lambda_1 \ll \Lambda_2$ there is a Type I factor $\mc{N}$ such that $\mc{R}_{\Lambda_1} \subset \mc{N} \subset \mc{R}_{\Lambda_2}$.
\end{definition}
The notation $\Lambda_1 \ll \Lambda_2$ means that $\Lambda_1 \subset \Lambda_2$ and that the edges of $\Lambda_1$ and $\Lambda_2$ are sufficiently far removed. For the models that we study in this paper it is sufficient to demand that there is no star or plaquette that has a non-empty intersection with both $\Lambda_1$ and $\Lambda_2$.

The approximate split property is a variant of the split property as it appears in algebraic quantum field theory~\cite{MR0345546,MR848392} and in operator algebra~\cite{MR735338}. The approximate split property also plays a role in the definition of a \emph{cone index} that tells us something about the number of superselection sectors the theory has~\cite{klindex}. There are nice physical consequences of the approximate split property: it implies a certain statistical independence of the regions $\Lambda_1$ and $\Lambda_2^c$. In particular one can find normal product states across these regions, so it is possible to find states which do not violate Bell's inequality~\cite{MR984150}. In fact one can \emph{locally} (in the sense that one acts only with operators in $\mc{R}_{\Lambda_1}$ or $\mc{R}_{\Lambda_2^c}$) such product states~\cite{MR895295}.

That the approximate split property holds in Kitaev's quantum double model can be seen as follows. From the proof of the uniqueness of the translational invariant ground state outlined in Section~\ref{sec:ground} one can see that the ground state $\omega_0$ is actually a product state, when one restricts to regions that are sufficiently far away. More concretely, let $\Omega$ be the GNS vector for $\omega_0$. We will write $\omega_0$ again for the state on $\mc{R}_{\Lambda_1} \vee \mc{R}_{\Lambda_2^c}$ induced by the vector $\Omega$. Note that it is normal, since it is a vector state for the von Neumann algebra. Note that we remarked before that $\omega_0$ is actually a product state for $\alg{A}(\Lambda_1)$ and $\alg{A}(\Lambda_2^c)$ if the boundaries of $\Lambda_1$ and $\Lambda_2$ are sufficiently far apart. This is precisely guaranteed by the condition $\Lambda_1 \ll \Lambda_2$. One can then show what $\omega_0(AB) = \omega_0(A) \omega_0(B)$ if $A \in \mc{R}_{\Lambda_1}$ and $B \in \mc{R}_{\Lambda_2^c}$. The approximate split property then follows from the same proof as given in~\cite[Thm. 5.2]{toricendo}.

Another way to prove the approximate split property is to explicitly construct a unitary that as in~\cite{haagdtoric}. We do not attempt a proof along these lines here, although we believe that using the techniques developed above for the proof of Haag duality, the proof carries over to the present situation without much changes. Indeed, the main idea behind the proof is to remove some of the ambiguity in the description of a vector in the form $F_1 \cdots F_n \Omega$ due to the invariance of states under ribbon deformations. This can be done using the same techniques as employed above. This explicit construction can be helpful in the calculation of the cone index in concrete examples~\cite{klindex}, but for our present purposes it is not necessary.

\section{Sector theory for abelian models}
As an application of Haag duality we outline the sector theory for abelian groups $G$, in the spirit of the Doplicher-Haag-Roberts programme in algebraic QFT~\cite{MR0297259,MR0334742}. The goal here is to retrieve all properties of the superselection sectors (or charges) in the theory, from a few basic principles. We will construct equivalence classes of such sectors for quantum double models for abelian groups $G$, and show explicitly how one can obtain the fusion and braiding rules. The techniques that we will use here were developed in~\cite{toricendo}, which essentially deals with the case $G = \mathbb{Z}_2$. The main ideas are the same in the case of general abelian $G$, hence we will focus here on those steps that are different.

The goal is to characterise ``single charge'' representations. These representations describe how the observables of the system change in the presence of a \emph{single} charge (or quasi-particle excitation) in the background. The different superselection sectors or charges correspond to equivalence classes of irreducible representations of $\alg{A}$~\cite{MR1405610}. This implies that vector states in inequivalent ``charged'' representations can not be coherently superposed. Alternatively one can see that by local operations one cannot transform a vector state in one such irreducible representation into a vector state of another (inequivalent) irreducible representation. Physically this means that one cannot change the total ``charge'' of the system with local operations. This is exemplified in the quantum double model by the property that the ribbon operators always create a pair of \emph{conjugate} charges, hence they do not change the \emph{total} charge of the system.

There are very many equivalence classes of irreducible representations, most of which do not carry any reasonable physical interpretation. It is therefore necessary to restrict the representations of interest. Recall that in the class of models that we are interested in, excitations or charges can be obtained from the ground state by applying a ribbon operator. Note that this always gives us a \emph{pair} of excitations if the model is defined on the plane. The idea is then to move one end of the ribbon (or physically, one of the charges) to infinity. For a related construction of charged states in $\mathbb{Z}_N$ Higgs models, see for example~\cite{MR1341694,MR728449}. The charge at the fixed endpoint can only be detected by measuring a ``Wilson loop'' that encloses the charge. Hence if we disallow operators that form a loop around the charge, it cannot be detected and the state will look like the ground state for such measurements.

What does this mean for the corresponding representations, obtained via the GNS construction for example? One can choose \emph{any} cone, and restrict to measurements outside such a cone. By the argument above this should look like the ground state representation. We therefore restrict to those representations that satisfy (c.f.~\cite{MR660538,toricendo})
\begin{equation}
	\pi_0 \upharpoonright \alg{A}(\Lambda^c) \cong \pi \upharpoonright \alg{A}(\Lambda^c)
	\label{eq:coneselect}
\end{equation}
for any cone $\Lambda$. That is, the representation $\pi$ is unitary equivalent to the ground state representation, but only when one restricts to observables \emph{outside} a cone. Equation~\eqref{eq:coneselect} is called a \emph{selection criterion}. The construction of such representations that we will outline below will make clear why this is a physically reasonable criterion. We stress that equation~\eqref{eq:coneselect} should hold for \emph{all} cones $\Lambda$ (where the unitary setting up the equivalence may depend on the cone).

By finding all representations that satisfy this criterion one finds a list of all charges that the system supports. But one can recover much more structure, and this is the point where Haag duality comes in: using Haag duality we can instead look at maps of $\alg{A}$ into a slightly bigger algebra, and in fact these maps can be extended to endomorphisms of this bigger algebra. To see this, fix a cone $\Lambda$ and let $V$ be the unitary such that $\pi_0(A) = V \pi(A) V^*$ for all $A \in \alg{A}(\Lambda^c)$. Then define $\alpha(A) = V \pi(A) V^*$ for all $A \in \alg{A}$. Then we have that, for $A \in \alg{A}(\Lambda)$ and $B \in \alg{A}(\Lambda^c)$,
\[
	\pi_0(B) \alpha(A) = V \pi(BA) V^* = V \pi(AB) V^* = \alpha(A) \pi_0(B).
\]
Hence by Haag duality it follows that $\alpha(A) \in \mc{R}_\Lambda$. As mentioned, $\alpha$ can be extended to a proper endomorphism. That is, one can introduce an auxiliary algebra $\alg{A}^{\Lambda_a}$ (where $\Lambda_a$ is a fixed cone), such that the map $\rho$ can be extended to an endomorphism of $\alg{A}^{\Lambda_a}$~\cite{MR660538,toricendo}. This is mainly a technical issue which we will large suppress here. In the explicit construction of such maps $\alpha$ below, it turns out that we can even restrict to \emph{automorphisms} of $\alg{A}$, although for the construction of braiding operators the extension to the auxiliary algebra is necessary. In addition, with Haag duality it follows that all the results that we show for these automorphisms are true for \emph{any} representative in the same equivalence class, even if it cannot be restricted to an automorphism of $\alg{A}$.

The advantage of using automorphisms or endomorphisms is that these can be composed, unlike representations. That is, we can define $\alpha \otimes \beta := \alpha \circ \beta$. The interpretation is that we first add a charge $\beta$, then a charge $\alpha$. In addition, if $S$ is an intertwiner from $\alpha_1$ to $\alpha_2$, meaning $S \alpha_1(A) = \alpha_2(A) S$ for all $A \in \alg{A}$, and $T$ is an intertwiner from $\beta_1 \to \beta_2$, it follows that $S \otimes T := S \alpha_1(T)$ is an intertwiner from $\alpha_1 \otimes \beta_1$ to $\alpha_2 \otimes \beta_2$. Using Haag duality and the extension of $\alpha_1$ to the auxiliary algebra one can show that this is well-defined. This makes the category of localised and transportable (which we will discuss below) endomorphisms, with as morphisms the intertwiners. Studying the superselection sectors is then studying the properties of this category. In this case this amounts to showing that it is in fact the representation category of the quantum double of the group $G$. This is a \emph{modular tensor category}~\cite{MR1797619}, as is appropriate for applications to quantum computing~\cite{MR2200691,Wang}. However note that we only consider abelian models at the moment, which from a quantum computation point of view are less interesting. We comment briefly on this point at the end of this paper.

\subsection{Construction of irreducible sectors}\label{sec:irred_secs}
The first task is to construct different equivalence classes of representations satisfying the selection criterion. We already mentioned that the ribbon operators create a pair of excitations. We will use this fact to first construct ``charged states'', from which the representations can be obtained straightforwardly. As expected, to each element $c \in G$ and irreducible representation $\chi$ of $G$ (that is, a character), we can associate an equivalence class of representations. To this end, fix a cone $\Lambda$ and consider a semi-infinite ribbon $\rho$ inside $\Lambda$. That is, one end of $\rho$ is fixed, the other end is thought of to be sent to infinity. The ribbon consisting of the first $n$ triangles will be denoted by $\rho_n$. We associate an endomorphism (in fact, since the model is abelian this will be an automorphism) to the each pair $(\chi,c)$ and semi-infinite ribbon $\rho$. In the next sub-section we will show that the choice of ribbon is not important, in the sense that another choice will lead to a unitarily equivalent automorphism.

The operators $F^{\chi,c}_\rho$ defined in equation~\eqref{eq:abribbon} create a pair with charge $(\chi,c)$ at the start of $\rho$ and its conjugate at the other end. Therefore one can think of the following map as describing the effect of the presence of this pair on an observable $A$:
\[
\alpha_\rho^{\chi,c}(A) := (\operatorname{Ad} F_\rho^{\chi,c})(A) = F_\rho^{\chi,c} A (F_\rho^{\chi,c})^*.
\]
Note that since $F_\rho^{\chi,c}$ is unitary this map is an automorphism. The idea is to take the limit in which we extend $\rho$ to infinity. The next proposition shows that this indeed works.
\begin{proposition}
	\label{prop:posmap}
Let $\rho$ be a ribbon extending to infinity, and denote $\rho_n$ for the ribbon consisting of the first $n$ triangles of $\rho$. Suppose that $(\chi,c)$ is as above. Then for each $A \in \alg{A}$ the limit
\begin{equation}
	\label{eq:posmap}
	\alpha(A) := \lim_{n \to \infty} \alpha_{\rho_n}^{\chi,c}(A)
\end{equation}
converges in norm and this defines an automorphism $\alpha: \alg{A} \to \alg{A}$. This map has the following properties:
\begin{enumerate}[(i)]
	\item \label{it:trivial}$\alpha(A) = A$ for $A \in \alg{A}$ with $\supp(A)$ disjoint from $\rho$;
	\item \label{it:local} If $A \in \alg{A}_{loc}$, then $\alpha(A) = \alpha_{\widehat{\rho}}^{\chi,c}(A)$ for any ribbon $\widehat{\rho} \subset \rho$ such that $\supp(A) \cap \rho \subset \widehat{\rho}$.
\end{enumerate}
The last property says that it is enough to move one end of the ribbon far enough away so that it is disjoint from the support of a local observable $A$.
\end{proposition}
\begin{proof}
	Let $A \in \alg{A}$ be a local operator. Then since $\rho$ goes to infinity, there is some $N$ such that $\supp(A) \cap (\rho \setminus \rho_n) = \emptyset$ for all $n > N$. In addition, from Lemma~\ref{lem:abdecompose} it follows that $F^{\chi,c}_{\rho_n} = F^{\chi,c}_{\rho_N} F^{\chi,c}_{\rho_n \setminus \rho_N}$. Because of locality and because $F^{\chi,c}_{\rho_N\setminus\rho_n}$ is unitary, it is clear that the limit converges in the operator norm, since the sequence $\alpha_{\rho_n}^{\chi,c}(A)$ is eventually constant. Note that this is essentially property~(\ref{it:local}). The maps are clearly bounded, hence by continuity they can be extended to a map $\alpha$ of $\alg{A}$. An easy check along the lines above shows that $\alpha(AB) = \alpha(A) \alpha(B)$ and $\alpha(A^*) = \alpha(A)^*$.

	We still have to show that $\alpha$ is an automorphism. The easiest way to do this is by constructing an inverse. Consider $(\overline{\chi}, \inv{c})$ where $\overline{\chi}$ is the complex conjugate of the character $\chi$, which again is a character. A simple calculation shows that $F^{\chi,c}_\rho F^{\overline{\chi},\inv{c}}_\rho = I$, and the same holds with the order reversed. It follows that $\overline{\alpha}$ which is defined in the same way as $\alpha$, but with the pair $(\overline{\chi}, \inv{c})$ satisfies $\overline{\alpha} \circ \alpha(A) = \alpha \circ \overline{\alpha}(A) = A$ for all local $A$, hence it is the inverse of $\alpha$.

Property (\ref{it:trivial}) immediately follows from locality.
\end{proof}

As we will see later the automorphisms constructed above give us representatives of the equivalence classes of representations satisfying the selection criterion. Anticipating this, we will also write $\alpha^{\chi,c}$ for an automorphism defined in such a way, or even $\alpha_\rho^{\chi,c}$ if we want to emphasize the ribbon to infinity. Note that each pair $(\chi,c)$ gives rise to many different automorphisms, since one can choose many different ribbons. If the ribbon is not important, we sometimes refer to any representative of this class of automorphisms by $(\chi,c)$. The following proposition shows that the automorphisms associated to different pairs $(\chi,c)$ belong to different superselection sectors, as expected. The idea behind the proof is that one can always detect the total charge in any finite region by pulling a charge and its conjugate from the vacuum, moving one charge around the region, and fusing again.
\begin{proposition}
	\label{prop:sectors}
If $(\sigma,c) \neq (\chi, d)$ then the corresponding localised automorphisms belong to different superselection sectors.
\end{proposition}
\begin{proof}
	Write $(\pi_0, \mc{H}, \Omega)$ for the GNS triple corresponding to the ground state $\omega_0$. Note that since $\omega_0$ is pure it follows that $\pi_0$ is irreducible. Because $\alpha^{\sigma,c}$ is an automorphism, $\pi_0 \circ \alpha^{\sigma,c}$ is also irreducible and $(\pi_0 \circ \alpha^{\sigma,c}, \mc{H}, \Omega)$ is a GNS triple for the state $\omega_0 \circ \alpha^{\sigma,c}$. A similar statement is of course true for the state $\omega_0 \circ \alpha^{\chi, d}$. To prove the claim it therefore suffices to show that the two states can be distinguished by an operator localised outside some arbitrary finite region $\mc{O}$ by Corollary 2.6.11 of~\cite{MR887100}. This is true because quasi-equivalent irreducible representations are unitarily equivalent.

	Now let $\mc{O}$ be any finite set. Then we can find a closed rotationally invariant ribbon $\widehat{\rho}$ encircling the region $\mc{O}$ and such that the endpoint of the ribbon $\rho$ that extends to infinity lies in the bounded area encircled by $\widehat{\rho}$. To this ribbon we associate the projection $K^{\sigma c}$, projecting onto the subspace of charge $(\sigma,c)$ in the region enclosed by $\widehat{\rho}$. It is defined as follows (c.f. equation~(B.75) of \cite{PhysRevB.78.115421}):
\[
K^{\sigma c}_{\widehat{\rho}} = \frac{1}{|G|} \sum_{g \in G} \overline{\sigma}(g) F^{g,c}_{\widehat{\rho}}.
\]
If $(\sigma,c) \neq (\chi, d)$ it follows that (by the discussion in Appendix B.9 of~\cite{PhysRevB.78.115421})
\[
\left| \omega_0 \circ \alpha^{\sigma,c}(K^{\sigma,c}_{\widehat{\rho}}) - \omega_0 \circ \alpha^{\chi,d}(K^{\sigma,c}_{\widehat{\rho}}) \right| = |1-0| \geq \frac{1}{2} \left\| K^{\sigma,c}_{\widehat{\rho}} \right\|.
\]
This completes the proof.
\end{proof}
In the next section we will show that the automorphisms are transportable, which will imply that the automorphisms defined on different ribbons, but with respect to the same pair $(\sigma,c)$ belong to the same sector.

\subsection{Transportability}\label{sec:transport}
Suppose that we have an automorphism $\alpha$ as defined above such that $\alpha$ is localised in a cone $\Lambda$. Then $\alpha$ is said to be transportable if for any cone $\Lambda'$, there is an automorphism $\beta$ localised in $\Lambda'$ such that $\alpha$ is unitarily equivalent to $\beta$. This unitary does not need to be in $\alg{A}$ (and generally also is not), but by Haag duality it follows that if $\widehat{\Lambda}$ is a cone containing both $\Lambda$ and $\Lambda'$, then any unitary $V$ setting up such an equivalence is contained in $\mc{R}_{\widehat{\Lambda}}$. Such a unitary will also be called a \emph{charge transporter}. We first show that the automorphisms are indeed transportable, and then give an explicit sequence that converges in the weak operator topology to a charge transporter. The proof largely follows the proof in the toric code case (up to some subtleties)~\cite{toricendo}, but since we need the construction to calculate the statistics, we recall the main line of argument.

Fix a pair $(\chi,c)$ and two semi-finite ribbons $\rho_i$, $i=1,2$, with corresponding automorphisms $\alpha_i$. First consider the case that both ribbons start at the same site. With Lemma~\ref{lem:indeprib} one can show that the states $\omega \circ \alpha_i$ are equal, by first showing equality on the dense subset of local observables. On the other hand, as was remarked in the proof of Proposition~\ref{prop:sectors}, both representations $\pi \circ \alpha_i$ are GNS representations for this state. Hence by the uniqueness of the GNS representation, the two are unitarily equivalent. Note that in addition we may assume that such a unitary intertwiner $V$ satisfies $V\Omega = \Omega$. Requiring this will fix an irrelevant phase.

Suppose now that the two ribbons do not start at the same site and that we consider a charge $(\chi,c)$. Then we get corresponding automorphisms $\alpha_1$ and $\alpha_2$. We can then extend the ribbon $\rho_1$ by a ribbon $\rho$, such that $\rho$ and $\rho_2$ start at the same site. This gives us an automorphism $\alpha_{\rho \rho_1}$, defined in terms of the extended ribbon, that is unitarily equivalent to $\alpha_2$, by the argument in the previous paragraph. The claim follows by noting that $\alpha_1$ and $\alpha_{\rho\rho_1}$ are unitarily equivalent. This can be seen because $F_{\rho}^{\overline{\chi},\inv{c}}$ is a unitary operator, and
\[
F_{\rho}^{\overline{\chi}, \inv{c}} \alpha_{\rho\rho_1}(A) (F_{\rho}^{\overline{\chi}, \inv{c}})^* = \alpha_1(A).
\]
This can be seen by noting that if a ribbon $\rho$ coincides with the first part of a ribbon $\widehat{\rho}$, then $F^{\chi,c}_\rho F^{\overline{\chi}, \inv{c}}_{\widehat{\rho}} = F^{\overline{\chi},\inv{c}}_{\widehat{\rho} \setminus \rho}$. This equality can be easily verified using equation~\eqref{eq:ribdecompose}.

For the calculation of the braiding rules of the anyons, which we will outline below, it is useful to have a more explicit description of the intertwiners setting up the equivalence. To this end we construct a sequence $V_n$ of unitaries converging to $V$ in the weak operator topology. For simplicity we again assume that the two semi-infinite ribbons $\rho_1$ and $\rho_2$ start at the same site. With $\rho^n_i$ we mean the finite ribbon consisting of the first $n$ triangles of the ribbon $\rho_i$. For each $n$, choose a ribbon $\widehat{\rho}_n$ from the site at the end of $\rho_1^n$ to the site at the end of $\rho_2^n$, in such a way that $\rho_1^n \widehat{\rho}_n$ is a ribbon and the distance of $\widehat{\rho}_n$ to the (fixed) starting point of $\rho_i$ goes to infinity as $n \to \infty$. This ensures that the ribbons $\widehat{\rho}_n$ will avoid any finite subset of the system when $n$ is large enough. By Lemma~\ref{prop:posmap} it also follows that for $n$ large enough, $\alpha_1^n(A) = F_{\rho_n}^{\chi, c} A F_{\rho_n}^{\overline{\chi}, \inv{c}}$ for $A$ strictly local.

Define $V_n = F_{\rho_2^n}^{\chi, c} F_{\rho_1^n \widehat{\rho}_n}^{\overline{\chi}, \inv{c}}$. The claim is that the sequence $V_n$ converges to $V$. Using the remark above about strictly local observables, a straightforward calculation shows that $V_n \alpha_1(A) = \alpha_2(A) V_n$ if $A$ is local and $n$ is big enough. Another remark is that using the techniques that we employed in the proof of Haag duality, it follows that $V_n \Omega = \Omega$. To see this, note that $F_{\rho_n'}^{\chi,c}$ and $F_{\rho_n}^{\overline{\chi},\inv{c}}$ create opposite charges at the endpoints of the ribbons. Since all charges are abelian, these opposite charges fuse to the vacuum. This can be seen explicitly by using that $F^{h,g}_{\rho} \Omega$ only depends on the endpoints of $\rho$, hence we can use this to change the path $\rho_n'$ to $\rho_n$ when acting on the ground state. Since $F^{\chi,c}_{\rho_n} F^{\overline{\chi},\inv{c}}_{\rho_n} = F^{\operatorname{id},e}_{\rho_n}$, the claim follows.

With these observations, we find for $A$ and $B$ strictly local operators and $n$ large enough, that
\[
\begin{split}
	\langle \alpha_1(A) &\Omega, V \alpha_1(B) \Omega \rangle = \langle \alpha_1(A) \Omega, \alpha_2(B) V \Omega \rangle = \\
		&\langle \alpha_1(A) \Omega, \alpha_1(B) V_n \Omega \rangle = \langle \alpha_1(A) \Omega, V_n \alpha_1(B) \Omega \rangle.
\end{split}
\]
Since $\alpha_1$ is an automorphism, it follows that the set $\alpha_1(A)$ for local operators $A$ is dense in the Hilbert space. Because the sequence $V_n$ is uniformly bounded, it follows that $V_n$ indeed converges to $V$. Note that if $\Lambda$ is a cone containing both ribbons $\rho_1$ and $\rho_2$ we can choose $V_n \in \alg{A}(\Lambda)$ and consequently $V \in \alg{A}(\Lambda)'' = \mc{R}_{\Lambda}$, as also follows from Haag duality.

The discussion so far can be summarised as the following theorem.
\begin{theorem}
	Let $G$ be a finite abelian group and let $\pi_0$ be the ground state representation of the quantum double model for $G$. Then for each pair $(\chi, c)$ where $\chi$ is a character of $G$ and $c \in G$, there is an equivalence class of representations satisfying the selection criterion~\eqref{eq:coneselect}. The representation $\pi_0 \circ \alpha$, where $\alpha$ is localised in some cone $\Lambda$ and constructed as above, is a representative of such an equivalence class. The equivalence classes corresponding to distinct pairs $(\chi,c)$ are disjoint.
\end{theorem}

\subsection{Fusion and statistics}
Fusion rules tell us what happens if we combine (``fuse'') to charges. More precisely, they give a decomposition of the tensor product $\alpha \otimes \beta$ of irreducible endomorphisms as a direct sum of such endomorphisms. The fusion rules are independent of the chosen representatives. Hence it suffices to fix a cone $\Lambda$ and a path $\rho$ to infinity inside this cone. We can then consider automorphisms $\alpha^{\chi,c}$ defined as above, acting along the ribbon $\rho$. Note that by Proposition~\ref{prop:posmap}(\ref{it:local}), for local observables it is enough to consider only finite parts $\rho_n$ of the path $\rho$. Note that for any finite ribbon $\xi$ we have $F^{\chi,c}_\rho F^{\sigma,d}_\rho = F^{\chi \sigma, cd}_\rho$ as was remarked after equation~\eqref{eq:abribbon}, where $\chi \sigma$ is the character obtained by pointwise multiplication. Hence we find the fusion rules
\[
\alpha^{\chi,c} \otimes \alpha^{\sigma, d} \cong \alpha^{\chi\sigma,cd}.
\]
Note that in particular we see that the conjugate charge of $(\chi,c)$ is $(\overline{\chi}, \overline{c})$.

To study the statistics we have to relate $\alpha \otimes \beta$ to $\beta \otimes \alpha$. For the construction we need to be able to talk about the relative position of two charges localised in cones. That is, we want to say that one cone is to the left of the other one. This can be done unambiguously by fixing an auxiliary cone: for convenience one can take the cone $\Lambda_a$ briefly mentioned above. Then we can define a relation $\Lambda_1 < \Lambda_2$ for two disjoint cones (see~\cite{toricendo} for details). This singling out of a particular direction is analogous to the technique of puncturing the circle in, for example, conformal field theory. Alternatively one can cover the lattice by different ``charts'' as in~\cite{MR1104414}.

\begin{figure}
    \centering
    \includegraphics[width=.5\linewidth]{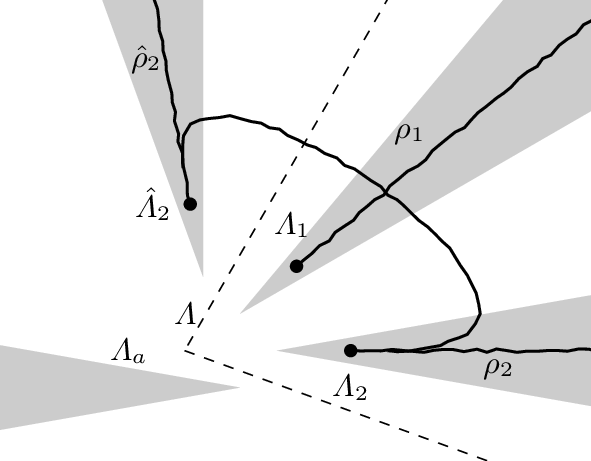}
	\caption{The choice of auxiliary cone $\Lambda_a$, as well as ribbons $\rho_1$ and $\rho_2$ that we use in the calculation of the braiding operators. The idea is to move the charge at the end of the ribbon $\rho_2$ to the end of the ribbon $\widehat{\rho}_2$, while the other charge stays in place.}
    \label{fig:braiding}
\end{figure}

Now suppose that we have two charges $\alpha$ and $\beta$ localised in cones $\Lambda_1$ and $\Lambda_2$. To construct a unitary $\varepsilon_{\alpha,\beta}$ intertwining  $\alpha \otimes \beta$ and $\beta \otimes \alpha$, first choose a cone $\widehat{\Lambda}_2$ to the \emph{left} of $\Lambda_1$ (and disjoint of it), see Fig.~\ref{fig:braiding}. Then there is an intertwiner  $V$ transporting the charge $\beta$ in $\Lambda_2$ to a charge $\widehat{\beta}$ in the cone $\widehat{\Lambda}_2$. Since $\widehat{\Lambda}_2$ and $\Lambda_1$ are disjoint, it follows by the localisation properties of the automorphisms that $\alpha \otimes \widehat{\beta} = \widehat{\beta} \otimes \alpha$. Finally, the charge $\widehat{\beta}$ can be transported back to the cone $\Lambda_2$. Note that the physical picture is precisely what one would think of as a braiding operation. This procedure leads to the following expression, which one can show depends only on the position of the cone $\widehat{\Lambda}_2$ relative to $\Lambda_1$, not on the specific choice of $V$ (c.f.~\cite{MR1016869}):
\[
	\varepsilon_{\alpha,\beta} = V^* \alpha(V).
\]
This unitary intertwines $\alpha \otimes \beta$ and $\beta \otimes \alpha$. One can show that it has all the properties a ``braiding'' should have (compare for example with~\cite{halvapp}).

Note that in the previous section we have constructed a sequence $V_n$ converging to $V$ in the weak operator topology. Since $\alpha$ can be extended to a weakly continuous map on the auxiliary algebra $\alg{A}^{\Lambda_a}$, we can calculate $\alpha(V) = \wlim_{n \to \infty} \alpha(V_n)$. We are interested here in the calculation of the modular matrix $S$, whose entries are in the present case given by $S_{\alpha,\beta} = \varepsilon_{\alpha,\beta} \circ \varepsilon_{\beta,\alpha}$. Note that because of irreducibility of $\alpha \otimes \beta$ this is an element of $\mathbb{C} I$ and hence can be identified with a scalar. It only depends on the equivalence classes of $\alpha$ and $\beta$, so that we can choose representatives in a convenient way. We do this as in Fig.~\ref{fig:braiding}: we choose two non-intersecting ribbons $\rho_i$ that can be localised in the same cone $\Lambda$. For the transported automorphisms we choose $\widehat{\rho}_1 = \rho_1$ and for $\widehat{\rho}_2$ a ribbon to the \emph{left} of the cone $\Lambda$, such that it is inside a cone $\widehat{\Lambda}_2$ that is disjoint from $\Lambda$. A sequence $V_n$ of charge transporters can then be constructed as in Section~\ref{sec:transport}, and it remains to calculate $\alpha(V_n)$.

This amounts to a straightforward application of the definitions. For convenience we can choose the ribbons connecting the $n$-th triangle of $\rho_2$ and $\widehat{\rho}_2$ in such a way that they cross the ribbon $\rho_1$ exactly once. Now note that for each $n > 0$, there is an integer $N(n)$ such that
\[
\alpha(V_n) = F_{(\rho_1)_{N(n)}}^{\chi_1,c} V_n (F_{(\rho_1)_{N(n)}}^{\chi_1,c})^*,
\]
by Proposition~\ref{prop:posmap}. Note that by construction $V_n$ is a product of two ribbon operators. Since the ribbon on which $V_n$ is defined and the ribbon $\rho_1$ cross exactly once, we can commute $V_n$ with the ribbon operator on the left of it in the expression above, at the expense of a phase according to equation~\eqref{eq:irrepbraid}. It follows that $\alpha(V) = \overline{\chi}_1(d) \overline{\chi}_2(c) V$ and hence $\varepsilon_{\alpha,\beta} = \overline{\chi}_1(d) \overline{\chi}_2(c) I$.

The operator $\varepsilon_{\beta,\alpha}$ can be found in the same way: we move the charge $\alpha$ to the left cone (and back). Since in this case the ribbons used in the construction of the appropriate intertwiner $W$ do not cross, it follows that $\varepsilon_{\beta,\alpha} = W^* \beta(W) = W^*W = I$. This gives us Verlinde's matrix $S$~\cite{MR954762}, whose entries are $S_{\alpha,\beta} = \varepsilon_{\alpha,\beta} \circ \varepsilon_{\beta,\alpha}$ in the special case that each sector is abelian (as is the case here). A more thorough discussion of $S$ in the context of the theory of superselection sectors can be found in~\cite{MR1147467,MR1128146}. In the end we obtain
\[
S_{(\chi_1, c),(\chi_2,d)} =  \overline{\chi_1}(d) \overline{\chi_2}(c).
\]
This is (up to a factor due to a different choice of normalization), precisely the matrix obtained in~\cite[Thm. 3.2.1]{MR1797619} for the representation category of the quantum double $\qdg$.

This is of course no coincidence. There is a correspondence between the superselection sectors constructed here and the finite dimensional representations of $\qdg$, seen as a Hopf algebra. It is well known that the irreducible representations of $\qdg$ are in one-one correspondence with pairs consisting of an equivalence class of $G$ and irreducible representations of the centraliser of a representative of this equivalence class (see for example~\cite{MR1797619,BeShoWha:2011}). In the present case of abelian groups this reduces to the pairs $(\chi,c)$. The fusion rules established above are precisely those obtained from the representation theory. With a little bit of work one can in fact show that the sector theory is completely determined by the representation theory of $\qdg$, where the $\qdg$-linear maps between finite dimensional representations correspond to intertwiners between the sectors constructed here. In the language of category theory, this can be phrased as stating that the category of localised endomorphisms and the category of finite dimensional representations of $\qdg$ are equivalent as braided fusion categories. With the help of the results above, the arguments are very similar to the toric code case~\cite{toricendo}, and hence we will not repeat them here. In any case, the upshot is that understanding the sector theory is the same as understanding the representation theory of $\qdg$ (a well studied subject), and that all physical properties of the excitations can be obtained by representation theory.

There is still a point that has not been answered, however. In principle, it may be the case that we have not constructed all sectors. That is, there may be irreducible representations that satisfy the selection criterion~\eqref{eq:coneselect}, but are not unitarily equivalent to one of the charged representations constructed above. The question if such additional charges exist can be answered by computing the Jones-Kosaki-Longo index for pairs of cones~\cite{klindex}. In essence one has to consider two disjoint cones $\Lambda_1 \cup \Lambda_2$ and the von Neumann algebra generated by the local observables in these cones, and in addition the algebra generated by the \emph{commutant} of everything in the \emph{complement} of the two cones. This algebra contains the von Neumann algebra generated by the observables inside the cones, but also charge transporters that move charges from one cone to the other. The Jones-Kosaki-Longo attaches a number to the relative size of these two algebras, and one can show that it is related to the quantum dimension of the charges. We expect that the proof given in the toric code case~\cite{klindex} can be extended to quantum double models for abelian $G$, by using the methods developed in this paper. The techniques used are very close to the proof of Haag duality given in Section~\ref{sec:haagd}, and therefore we do not attempt to give a full proof here. In any case, we expect that this index is equal to $|G|^2$, and that we have therefore found all charges in the model.

\section{The case of non-abelian groups}
Many of the results so far are only proven for abelian groups $G$. A natural question is if the same results hold for \emph{non-abelian} groups. Such models are particularly interesting because they have non-abelian anyons, for which the braiding operators are not just a phase, but in general give rise to higher dimensional representations of the braid group. In Kitaev's model this can be exploited to implement unitary operations (\emph{gates}) from which quantum circuits can be build. Such quantum circuits perform quantum computation tasks. Under suitable conditions on the group $G$ Kitaev's model is in fact \emph{universal}, meaning that in principle any quantum computation algorithm could be implemented on top of the model~\cite{PhysRevA.67.022315,PhysRevA.69.032306}.

We believe that these non-abelian models can be studied along the same lines as the abelian ones. From a technical point of view, however, the analysis is much more involved. The difficulties mainly stem from the fact that for non-abelian $G$ the quantum double $\qdg$ has higher dimensional irreducible representations. This has a few consequences. First of all, rather than a \emph{single} ribbon operator being associated with a certain irreducible representation (such as the $F^{\chi,c}_\rho$ we used above), one has to deal with ``multiplets'' of ribbon operators, see for example equation~(B.66) of~\cite{PhysRevB.78.115421}. When acting with the operators in such a multiplet on the ground state $\Omega$, one can span a finite dimensional vector space. The star and plaquette operators at one of the endpoints of the ribbon give a natural action of the quantum double on this finite dimensional vector space, which then transforms as an irreducible representation under this action.

The second point is related to the tensor product of two irreducible representations. In the abelian case such a tensor product was again irreducible and of the same form. This is no longer true in general in the non-abelian case. Indeed, there are tensor products of representations that are the direct sum of more than one irreducible representation. On the level of the ribbon operators this has, for example, the consequence that when we multiply two ribbon operators acting along the same ribbon, in general it is not of the form of a single ribbon operator any more. This naturally makes the analysis more complicated. In addition the interchange of two ribbon operators is more complicated than just introducing a phase. Nevertheless, the representation theory of the quantum double is well understood, so we expect that the main ideas in our proof can be transferred to the non-abelian case. In particular, one should be able to use this knowledge of the representation theory to study the commutation properties of the ribbon operators, which are essential in the proof of Haag duality.

The non-abelianness also makes it more difficult to explicitly construct representatives of the charged sectors: instead of automorphisms one has to deal with endomorphisms and we cannot just conjugate with the ribbon operators to define them. Instead, one way would be to use \emph{amplimorphisms}, which are nothing but morphisms $\rho : \alg{A} \to M_n(\alg{A})$, the $n$-by-$n$ matrices with entries in $\alg{A}$. Such methods have been employed to describe localised (in intervals) charges in quantum spin systems on the line~\cite{MR1463825,MR1234107}. Unlike in the case of finitely localised excitations, in the case of conelike localisation we expect to be able to obtain proper endomorphisms again. One way to do this is to note that the cone algebras are infinite factors. This allows us to find isometries $V_i$ ($i=1,\dots, n$) in the cone algebra whose ranges sum up to the identity projection. In this way we can identify $\mc{H}$ with $\oplus_{i=1}^n \mc{H}$, and obtain an identification of an amplimorphism $\rho$ as above with an endomorphism of the cone algebra. This should make it possible to carry over the well-known structure of the amplimorphisms to endomorphisms, and build up representatives of each sector and find the braiding operators. We hope to return to this issue in the future.

\vspace{\baselineskip}
\textbf{Acknowledgements:} LF is supported by the European Research Council (ERC) through the Discrete Quantum Simulator (DQSIM) project. PN is supported by the Dutch Organisation for Scientific Research (NWO) through a Rubicon grant and partly through the EU project QFTCMPS and the cluster of excellence EXC 201 Quantum Engineering and Space-Time Research.

\bibliographystyle{abbrv}
\bibliography{refs}

\begin{thebibliography}{10}

\bibitem{MR2345476}
R.~Alicki, M.~Fannes, and M.~Horodecki.
\newblock A statistical mechanics view on {K}itaev's proposal for quantum
  memories.
\newblock {\em J. Phys. A}, 40(24):6451--6467, 2007.

\bibitem{MR1797619}
B.~Bakalov and A.~Kirillov, Jr.
\newblock {\em Lectures on tensor categories and modular functors}, volume~21
  of {\em University Lecture Series}.
\newblock American Mathematical Society, Providence, RI, 2001.

\bibitem{MR1341694}
J.~C.~A. Barata and F.~Nill.
\newblock Electrically and magnetically charged states and particles in the
  {$(2+1)$}-dimensional {$Z_N$}-{H}iggs gauge model.
\newblock {\em Comm. Math. Phys.}, 171(1):27--86, 1995.

\bibitem{BeShoWha:2011}
S.~Beigi, P.~W. Shor, and D.~Whalen.
\newblock The quantum double model with boundary: Condensations and symmetries.
\newblock {\em Communications in Mathematical Physics}, 306:663--694, 2011.

\bibitem{MR0438944}
J.~J. Bisognano and E.~H. Wichmann.
\newblock On the duality condition for quantum fields.
\newblock {\em J. Mathematical Phys.}, 17(3):303--321, 1976.

\bibitem{PhysRevB.78.115421}
H.~Bombin and M.~A. Martin-Delgado.
\newblock Family of non-{A}belian {K}itaev models on a lattice: {T}opological
  condensation and confinement.
\newblock {\em Phys. Rev. B}, 78(11):115421, Sep 2008.

\bibitem{MR887100}
O.~Bratteli and D.~W. Robinson.
\newblock {\em Operator algebras and quantum statistical mechanics. 1}.
\newblock Texts and Monographs in Physics. Springer-Verlag, New York, second
  edition, 1987.

\bibitem{MR1441540}
O.~Bratteli and D.~W. Robinson.
\newblock {\em Operator algebras and quantum statistical mechanics. 2}.
\newblock Texts and Monographs in Physics. Springer-Verlag, Berlin, second
  edition, 1997.

\bibitem{MR0345546}
D.~Buchholz.
\newblock Product states for local algebras.
\newblock {\em Comm. Math. Phys.}, 36:287--304, 1974.

\bibitem{MR848392}
D.~Buchholz, S.~Doplicher, and R.~Longo.
\newblock On {N}oether's theorem in quantum field theory.
\newblock {\em Ann. Physics}, 170(1):1--17, 1986.

\bibitem{MR660538}
D.~Buchholz and K.~Fredenhagen.
\newblock Locality and the structure of particle states.
\newblock {\em Comm. Math. Phys.}, 84(1):1--54, 1982.

\bibitem{MR1147468}
D.~Buchholz, G.~Mack, and I.~Todorov.
\newblock Localized automorphisms of the {${\rm U}(1)$}-current algebra on the
  circle: an instructive example.
\newblock In {\em The algebraic theory of superselection sectors ({P}alermo,
  1989)}, pages 356--378. World Sci. Publ., River Edge, NJ, 1990.

\bibitem{MR1128130}
R.~Dijkgraaf, V.~Pasquier, and P.~Roche.
\newblock Quasi {H}opf algebras, group cohomology and orbifold models.
\newblock {\em Nuclear Phys. B Proc. Suppl.}, 18B:60--72, 1991.
\newblock Recent advances in field theory (Annecy-le-Vieux, 1990).

\bibitem{MR0297259}
S.~Doplicher, R.~Haag, and J.~E. Roberts.
\newblock Local observables and particle statistics. {I}.
\newblock {\em Comm. Math. Phys.}, 23:199--230, 1971.

\bibitem{MR0334742}
S.~Doplicher, R.~Haag, and J.~E. Roberts.
\newblock Local observables and particle statistics. {II}.
\newblock {\em Comm. Math. Phys.}, 35:49--85, 1974.

\bibitem{MR735338}
S.~Doplicher and R.~Longo.
\newblock Standard and split inclusions of von {N}eumann algebras.
\newblock {\em Invent. Math.}, 75(3):493--536, 1984.

\bibitem{MR934283}
V.~G. Drinfel'd.
\newblock Quantum groups.
\newblock In {\em Proceedings of the {I}nternational {C}ongress of
  {M}athematicians, {V}ol. 1, 2 ({B}erkeley, {C}alif., 1986)}, pages 798--820,
  Providence, RI, 1987. Amer. Math. Soc.

\bibitem{MR728449}
K.~Fredenhagen and M.~Marcu.
\newblock Charged states in {$Z_{2}$} gauge theories.
\newblock {\em Comm. Math. Phys.}, 92(1):81--119, 1983.

\bibitem{MR1016869}
K.~Fredenhagen, K.-H. Rehren, and B.~Schroer.
\newblock Superselection sectors with braid group statistics and exchange
  algebras. {I}.\ {G}eneral theory.
\newblock {\em Comm. Math. Phys.}, 125(2):201--226, 1989.

\bibitem{MR1104414}
J.~Fr{\"o}hlich and F.~Gabbiani.
\newblock Braid statistics in local quantum theory.
\newblock {\em Rev. Math. Phys.}, 2(3):251--353, 1990.

\bibitem{MR1405610}
R.~Haag.
\newblock {\em Local quantum physics: Fields, particles, algebras}.
\newblock Texts and Monographs in Physics. Springer-Verlag, Berlin, second
  edition, 1996.

\bibitem{halvapp}
H.~Halvorson.
\newblock Algebraic quantum field theory.
\newblock In J.~Butterfield and J.~Earman, editors, {\em Philosophy of
  Physics}, pages 731--922. Elsevier, 2006.

\bibitem{HeRo:1970}
E.~Hewitt and K.~A. Ross.
\newblock {\em Abstract Harmonic Analysis. Vol. II: Structure and Analysis on
  Locally Compact Groups; Analysis on Locally Compact Abelian Groups}.
\newblock Springer, 1970.

\bibitem{MR1321145}
C.~Kassel.
\newblock {\em Quantum groups}, volume 155 of {\em Graduate Texts in
  Mathematics}.
\newblock Springer-Verlag, New York, 1995.

\bibitem{MR1951039}
A.~Kitaev.
\newblock Fault-tolerant quantum computation by anyons.
\newblock {\em Ann. Physics}, 303(1):2--30, 2003.

\bibitem{MR2200691}
A.~Kitaev.
\newblock Anyons in an exactly solved model and beyond.
\newblock {\em Ann. Physics}, 321(1):2--111, 2006.

\bibitem{PhysRevA.67.022315}
C.~Mochon.
\newblock Anyons from nonsolvable finite groups are sufficient for universal
  quantum computation.
\newblock {\em Phys. Rev. A}, 67(2):022315, Feb 2003.

\bibitem{PhysRevA.69.032306}
C.~Mochon.
\newblock Anyon computers with smaller groups.
\newblock {\em Phys. Rev. A}, 69(3):032306, Mar 2004.

\bibitem{toricendo}
P.~Naaijkens.
\newblock Localized endomorphisms in {K}itaev's toric code on the plane.
\newblock {\em Rev. Math. Phys.}, 23(4):347--373, 4 2011.

\bibitem{phdnaaijkens}
P.~Naaijkens.
\newblock {\em Anyons in infinite quantum systems: {QFT} in {$d=2+1$} and the
  toric code}.
\newblock PhD thesis, Radboud Universiteit Nijmegen, 2012.

\bibitem{haagdtoric}
P.~Naaijkens.
\newblock Haag duality and the distal split property for cones in the toric
  code.
\newblock {\em Lett. Math. Phys.}, 101(3):341--354, 2012.

\bibitem{klindex}
P.~Naaijkens.
\newblock Kosaki-{L}ongo index and classification of charges in {2D} quantum
  spin models.
\newblock {\em J. Math. Phys.}, 54:081901, 2013.

\bibitem{MR1463825}
F.~Nill and K.~Szlach{\'a}nyi.
\newblock Quantum chains of {H}opf algebras with quantum double cosymmetry.
\newblock {\em Comm. Math. Phys.}, 187(1):159--200, 1997.

\bibitem{MR2174961}
R.~Oeckl.
\newblock {\em Discrete gauge theory}.
\newblock Imperial College Press, London, 2005.
\newblock From lattices to TQFT.

\bibitem{MR1147467}
K.-H. Rehren.
\newblock Braid group statistics and their superselection rules.
\newblock In D.~Kastler, editor, {\em The algebraic theory of superselection
  sectors ({P}alermo, 1989)}, pages 333--355. World Sci. Publ., River Edge, NJ,
  River Edge, NJ, 1990.

\bibitem{MR1128146}
K.-H. Rehren.
\newblock Markov traces as characters for local algebras.
\newblock {\em Nuclear Phys. B Proc. Suppl.}, 18(2):259--268, 1991.
\newblock Recent advances in field theory (Annecy-le-Vieux, 1990).

\bibitem{RiDa:1975}
M.~A. Rieffel and A.~van Daele.
\newblock The commutation theorem for tensor products of von {N}eumann
  algebras.
\newblock {\em Bulletin of the London Mathematical Society}, 7(3):257--260, Nov
  1975.

\bibitem{MR984150}
S.~J. Summers and R.~Werner.
\newblock Maximal violation of {B}ell's inequalities for algebras of
  observables in tangent spacetime regions.
\newblock {\em Ann. Inst. H. Poincar\'e Phys. Th\'eor.}, 49(2):215--243, 1988.

\bibitem{MR1234107}
K.~Szlach{\'a}nyi and P.~Vecserny{\'e}s.
\newblock Quantum symmetry and braid group statistics in {$G$}-spin models.
\newblock {\em Comm. Math. Phys.}, 156(1):127--168, 1993.

\bibitem{Takesaki:2002}
M.~Takesaki.
\newblock {\em Theory of Operator Algebra I}, volume 124 of {\em Encyclopaedia
  of Mathematical Sciences}.
\newblock Springer Berlin / Heidelberg, 2002.

\bibitem{MR954762}
E.~Verlinde.
\newblock Fusion rules and modular transformations in {$2$}{D} conformal field
  theory.
\newblock {\em Nuclear Phys. B}, 300(3):360--376, 1988.

\bibitem{Wang}
Z.~Wang.
\newblock {\em Topological Quantum Computation}, volume 112 of {\em CBMS
  Regional Conference Series in Mathematics}.
\newblock Published for the Conference Board of the Mathematical Sciences,
  Washington, DC, 2010.

\bibitem{MR895295}
R.~Werner.
\newblock Local preparability of states and the split property in quantum field
  theory.
\newblock {\em Lett. Math. Phys.}, 13(4):325--329, 1987.

\end{thebibliography}

\end{document}